\documentclass[headings=standardclasses]{scrartcl}

\usepackage[T1]{fontenc}
\usepackage{lmodern}
\usepackage[british]{babel}
\usepackage{authblk}
\usepackage{amsmath}
\usepackage{amssymb}
\usepackage{subfiles}
\usepackage{amsthm}
\usepackage{mathtools}
\usepackage{physics}
\usepackage{hyperref} 
\usepackage{cleveref}
\usepackage{xr-hyper}
\usepackage{svg}
\usepackage{siunitx}
\usepackage{algorithm}
\usepackage{algorithmic}
\usepackage{booktabs}
\usepackage{makecell}
\usepackage{multirow}
\usepackage{csquotes} 
\usepackage[backend=biber, style=ieee]{biblatex}
\usepackage{graphicx}
\usepackage[]{todonotes}
\usetikzlibrary{shapes, arrows, positioning, calc}
\usepackage{upgreek}

\newcommand\e{\mathord{\,\mathrm{e}}}
\renewcommand\j{\mathord{\,\mathrm{j}\mkern1mu}}

\DeclareMathOperator*{\argmin}{arg\,min}

\DeclareMathOperator{\arctantwo}{arctan2}

\newcommand{\mat}[1]{\mathbf{#1}}
\renewcommand{\vec}[1]{\boldsymbol{\mathrm{#1}}}% NatCom

\newtheorem{theorem}{Theorem}

\newtheorem{corollary}{Corollary}

\newcommand{\printtitle}[2]{%
  \clearpage
  \begin{center}
  {\LARGE\bfseries #1\par}
  \vspace{1em}
  {\large #2\par}
  \vspace{1em}
  \end{center}
}

\begin{document}

    \title{Optimal Calibration of the Endpoint-corrected Hilbert Transform}
    \author{Eike Osmers\thanks{Corresponding author: \href{mailto:eike.osmers@tu-berlin.de}{eike.osmers@tu-berlin.de}}~}
    \author{Dorothea Kolossa}
    \affil{\parbox{\linewidth}{\centering
    		Technische Universität Berlin\protect\\
    		Faculty IV - Electrical Engineering and Computer Science\protect\\
    		Institute of Energy and Automation Technology\protect\\
    		Electronic Systems of Medical Engineering}}
    
    \date{\today}

	\maketitle
	
	\begin{abstract}
		\noindent Accurate, low-latency estimates of the instantaneous phase of oscillations are essential for closed-loop sensing and actuation, including (but not limited to) phase-locked neurostimulation and other real-time applications. The endpoint-corrected Hilbert transform (ecHT) reduces boundary artefacts of the Hilbert transform by applying a causal narrow-band filter to the analytic spectrum. This improves the phase estimate at the most recent sample. Despite its widespread empirical use, the systematic endpoint distortions of ecHT have lacked a principled, closed-form analysis. In this study, we derive the ecHT endpoint operator analytically and demonstrate that its output can be decomposed into a desired positive-frequency term (a deterministic complex gain that induces a calibratable amplitude/phase bias) and a residual leakage term that sets an irreducible variance floor. This yields (i) an explicit characterisation and bounds for endpoint phase/amplitude error, (ii) a mean-squared-error-optimal scalar calibration, and (iii) practical design rules relating window length, filter bandwidth and order, and centre-frequency mismatch to residual bias via an endpoint group delay. The resulting calibrated ecHT achieves near-zero mean phase error and remains computationally compatible with real-time pipelines. Code and analyses are provided at \url{https://github.com/eikeosmers/cecHT}.
	\end{abstract}
	
	\section*{Introduction}
	
	%\todo[inline]{https://www.nature.com/documents/ncomms-submission-guide.pdf}
	
	Real-time access to the phase of neural oscillations enables causal experiments and closed-loop stimulation, where stimuli are synchronised with a rhythm's phase. In neuroscience, this capability is central to closed-loop stimulation paradigms, in which the timing relative to ongoing brain rhythms can be controlled experimentally~\cite{grossAnalyticalMethodsExperimental2014,zrennerClosedLoopBrainStimulation2024}. More broadly, the same technical requirement arises in any real-time system that must estimate the instantaneous phase of a narrow-band process within tight latency constraints, such as phase synchronisation in power electronics, rotor phase estimation in electromechanical control systems, or carrier phase tracking in navigation and communication receivers. In all such settings, the estimator must produce a phase estimate in real time, with latency small compared to the signal's timescale, while the signal is still evolving.\\
	
	The vanilla approach is to construct the analytic signal using the Hilbert transform (HT) and take its argument to determine the instantaneous phase. In discrete-time implementations, this is typically achieved using a discrete Fourier transform (DFT)-based Hilbert transform applied to a short, sliding window. The DFT implicitly assumes that this window repeats periodically. However, when the signal is not exactly periodic across the window, as depicted in Fig.~\ref{fig:intro}, there is a discontinuity between the last and first sample. The Fourier representation of this discontinuity produces Gibbs ringing (overshoot and oscillations), which leads to large amplitude and phase distortions at the window boundaries, as seen in Fig.~\hyperref[fig:intro]{\ref*{fig:intro}f}. As real-time systems typically use only the most recent sample in the window (i.e.~the last one), the Gibbs phenomenon directly corrupts the most important estimate.\\
	
	\begin{figure}[!h]
		\centering
		\includegraphics[width=\columnwidth, trim={0 3em 0 1em}]{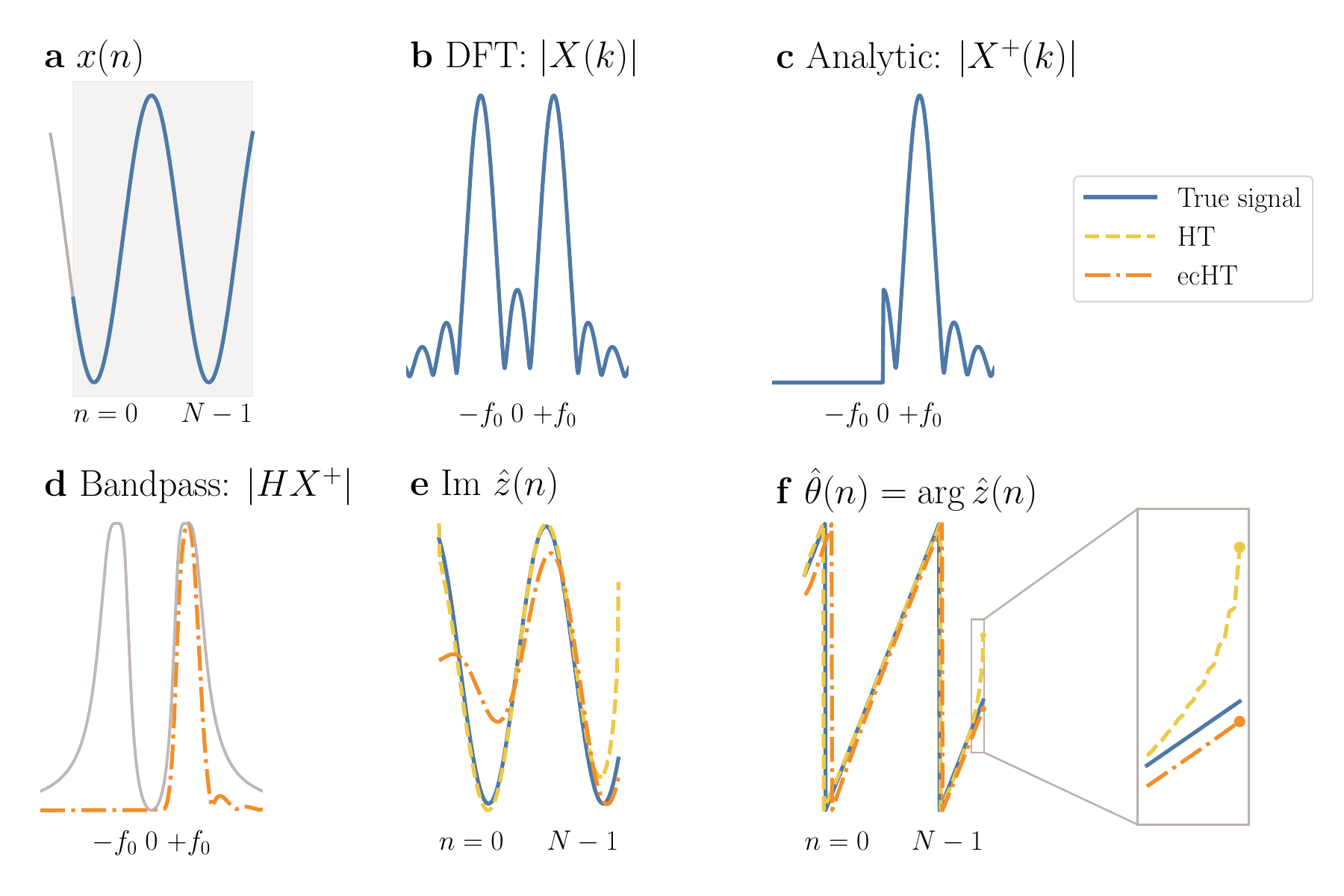}
		\caption{\textbf{Visual concept of the ecHT algorithm.} \textbf{a} Finite window cut from a longer signal $x$; the DFT assumes periodic repetition, creating a discontinuity at the boundaries. \textbf{b} DFT magnitude $\abs{X}$ with positive and negative sidelobes. \newline
		\textbf{c} Analytic signal spectrum $X^+$ after Hilbert mask removes negative frequencies. \textbf{d} The ecHT multiplies $X^+$ with a narrow-band Butterworth filter $H$ centred at $f_0$. \textbf{e} Time-domain imaginary parts: classical Hilbert transform (yellow) shows endpoint ringing; ecHT (orange) is closer to the true signal (blue). \textbf{f} Instantaneous phase near the endpoint: the ecHT reduces the phase error of the online Hilbert transform at the last sample but has itself a bias.}
		\label{fig:intro}
	\end{figure}
	
	The endpoint-corrected Hilbert Transform (ecHT) was introduced by Schregl\-mann et al.~\cite{schreglmannNoninvasiveSuppressionEssential2021} as a modification of the classical HT for such online use. The basic idea is to apply a causal, narrow-band bandpass filter to the analytic spectrum centred at the frequency of interest. In the time domain, this filtering reduces the high-frequency components associated with the sharp transition between successive windows. This effectively smooths the periodic continuation of the signal and minimises the Gibbs-induced ringing at the window edges. As the filter is causal and evaluated in a sliding fashion, the ecHT provides a recursive, sample-by-sample phase estimate, with the phase at the last sample of each window being the quantity of interest.\\
	
	One might try to mitigate boundary artefacts using simpler techniques such as zero-padding or windowing. However, zero-padding moves the discontinuity to the transition between data and padding, and does not systematically reduce the endpoint error. Tapered windows (e.g.~Hann) reduce spectral leakage by attenuating the edges of the window; however, in an online setting, they downweight the most recent sample. Symmetric windows or forward-backward filtering can greatly reduce phase distortion, but they require access to future samples and are therefore acausal.\\
	
	More recently, the ecHT has been applied to the real-time tracking of phases in a variety of biological signals, including essential tremor and alpha or theta rhythms. In these applications, the ecHT has enabled phase-locked stimulation using wearable EEG systems, and has thus served to demonstrate phase-dependent modulation of evoked responses during wakefulness and sleep~\cite{bresslerWearableEEGSystem2023,bresslerRandomizedControlledTrial2024,harlowIndividualizedClosedLoopAcoustic2024,hebronClosedloopAuditoryStimulation2024,jaramilloClosedloopAuditoryStimulation2024,vinao-carlJustPhaseCausal2024}. Busch et al.~\cite{buschRealtimePhaseAmplitude2022} observed that the ecHT introduces a distinctive phase shift and Liufu et al.~\cite{liufuOptimizingRealtimePhase2025} conducted a comprehensive evaluation of the ecHT across various signals. Their study revealed that the ecHT consistently delivers the most accurate and precise phase-locked outputs among several real-time methods. They also found that performance is strongly influenced by signal-to-noise ratio (SNR), amplitude and frequency variability, and window length.\\
	
	Beyond the ecHT, a variety of other real-time phase estimators have been proposed, including adaptive filtering approaches~\cite{zrennerRealtimeEEGdefinedExcitability2018,shirinpourExperimentalEvaluationMethods2020}, waveform-based and spectrogram-based methods~\cite{madsenNoTracePhase2019,shirinpourExperimentalEvaluationMethods2020,kimEEGPhaseCan2023}, state-space models~\cite{onojimaStateinformedStimulationApproach2021,wodeyarStateSpaceModeling2021,wodeyarDifferentMethodsEstimate2023,makarovaHardwareenabledLowLatency2025}, resonance-based loops~\cite{santostasiPhaselockedLoopPrecisely2016,rosenblumRealtimeEstimationPhase2021,buschRealtimePhaseAmplitude2022,choDevelopmentEvaluationRealTime2024}, and learned predictors~\cite{mcintoshEstimationPhaseEEG2020,liuNovelRealtimePhase2025}. Comparative studies suggest that many of these methods can achieve reliable phase tracking under appropriate conditions~\cite{shirinpourExperimentalEvaluationMethods2020,fersterBenchmarkingRealTimeAlgorithms2022}. State-space and resonance-based approaches can provide good tracking and handle several non-stationarities, but they typically require model fitting or careful parameter tuning. Spectral and time-frequency methods offer high flexibility across frequencies; however, their finite analysis window introduces a fundamental trade-off between temporal and frequency resolution. Waveform-based predictors also face real-time constraints: filter design and window length must be balanced to achieve an optimal combination of computational speed, latency, and phase-estimation accuracy. Learned predictors can adapt flexibly, but they require substantial training data and are more difficult to analyse theoretically. In this context, ecHT is a low-complexity, FFT-based method that is easy to implement and run in real time.\\
	
	%Despite its widespread use, previous research into the ecHT has largely been empirical. While existing studies have validated the ecHT experimentally, they have not provided a systematic derivation of its spectral properties, endpoint errors, or calibration strategies. Here, we address this issue. We provide an analytical DFT-based characterisation of the endpoint-corrected Hilbert transform, express its error analytically, and demonstrate how a single scalar can deliver a mean-squared-error (MSE)-optimal calibration. We refer to the calibrated variant of the ecHT as the \emph{cecHT}. We then translate these results into practical design rules concerning window length, filter bandwidth and order, and centre-frequency estimation. These rules allow the ecHT to achieve real-time phase estimates with mean errors of practically zero degrees under realistic conditions, as confirmed in simulations and real data.
	
	Here, we show that the endpoint-corrected Hilbert transform can be analytically decomposed, separating systematic bias from irreducible phase-dependent errors. Despite its widespread use, previous research into the ecHT has largely been empirical. While existing studies have validated the ecHT experimentally, they have not provided a systematic derivation of its spectral properties, endpoint errors, or calibration strategies. We provide an analytical DFT-based characterisation of the ecHT, express its error analytically, and demonstrate how a single scalar can deliver a mean-squared-error (MSE)-optimal calibration. We refer to the calibrated variant of the ecHT as the \emph{cecHT}. We then translate these results into practical design rules concerning window length, filter bandwidth and order, and centre-frequency estimation. These rules allow the ecHT to achieve real-time phase estimates with mean errors of practically zero degrees under realistic conditions, as confirmed in simulations and real data.
	
	\section*{Results}
	
	\subsection*{Optimal calibration}
	
	We show in Supplementary Theorem~\ref{thm:scalar-calibration} under 
	%the 
	mild assumptions 
	%in Supplementary Section~\ref{sec:scalar-calibration-general-proof} 
	that for any ecHT implementation, there exists a unique complex scalar $C_\text{opt}$ that minimises the mean-squared error
	\begin{align}
		J(C) = \mathbb{E}\left[\,\abs*{C\hat Z - Z}^2\,\right]
	\end{align}
	between the calibrated ecHT output and the ideal analytic signal. $C_\text{opt}$ is obtained by
	\begin{align}
		C_{\mathrm{opt}} &= \frac{\mathbb{E}[\hat Z^* Z]}{\mathbb{E}[\abs*{\smash{\hat Z}}^2]},
		\label{eq:Copt-general}\quad \text{with}\\
		J(C_{\mathrm{opt}}) &= \mathbb{E}[\abs{Z}^2]
		- \frac{\abs*{\mathbb{E}[\hat Z^{\ast} Z]}^2}{\mathbb{E}[\abs*{\hat Z}^2]}
		\;\le\; J(1) = \mathbb{E}[\abs*{\hat Z - Z}^2],
		\label{eq:Jmin-general-main}
	\end{align}
	where $Z$ is a random variable capturing the true endpoint and $\hat Z^*$ denotes the complex-conjugate ecHT estimate. Geometrically, multiplication by $C_\text{opt}$ is an orthogonal projection of $Z$ onto the one-dimensional subspace spanned by $\hat Z$ in the Hilbert space of complex random variables; it is the scalar Wiener filter that extracts the maximally linearly recoverable information determined by the squared magnitude of the correlation coefficient
	\begin{align}
		\rho_{Z\hat Z}
		= \frac{\mathbb{E}[\hat Z^{\ast} Z]}
		{\sqrt{\mathbb{E}[\abs*{\smash{\hat Z}}^2]\mathbb{E}[\abs{Z}^2]}},
	\end{align}
	and substituting into Eq.~\eqref{eq:Jmin-general-main} gives
	\begin{align}
		J(C_{\mathrm{opt}}) &= \mathbb{E}[\abs{Z}^2]\left(1 - \abs{\rho_{Z\hat Z}}^2\right).
	\end{align}
	Correlation (not calibration) sets the fundamental accuracy limit of the cecHT, and all sources of distortion impact performance only insofar as they reduce this correlation.\\
	
	For the narrow-band single-tone case of Eq.~\eqref{eq:x(n)}, this result specialises to
	\begin{align}
		C_{\mathrm{opt}} = \frac{G_+^*}{\abs{G_+}^2 + \abs{G_-}^2},
		\label{eq:Copt-single-tone}
	\end{align}
	which depends only on the positive-frequency gain $G_+$ and the negative-frequency leakage $G_-$ (c.f.~Supplementary Theorem~\ref{thm:single-tone-calibration}). Both quantities can be pre-computed analytically for any choice of window length, filter parameters, and centre frequency, requiring no data-driven fitting or iterative optimisation.\\
	
	While calibration eliminates the bias encoded in $G_+$, it cannot remove errors that arise from the leakage-dependent ripple $G_- \e^{-\j2\varphi_0}$. After calibration, the minimal achievable mean-squared error is
	\begin{align}
		J(C_{\mathrm{opt}}) = \frac{\abs{G_-}^2}{\abs{G_+}^2 + \abs{G_-}^2}.
		\label{eq:Jmin-single-tone}
	\end{align}
	This bound is deterministic and cannot be improved by averaging or post-processing as it reflects the fundamental information loss induced by finite-window spectral leakage.
	
	%	In realistic settings, the ecHT input contains both signal and additive noise $\eta$ with variance $\sigma^2_\eta$, the ecHT estimate becomes
	%	\begin{align}
		%		\hat Z = FZ + W,
		%	\end{align}
	%	where $W$ is the filtered noise contribution with variance $\sigma^2_W = G_\text{noise} \sigma^2_\eta$. The noise gain $G_\text{noise}$ can be precomputed from the ecHT impulse response.
	%	The deterministic calibration factor Eq.~\eqref{eq:Copt-single-tone} remains near-optimal at moderate SNR. The output SNR scales linearly with input SNR
	%	\begin{align}
		%		\operatorname{SNR}_\text{out} = G_\text{SNR} \operatorname{SNR}_\text{in},
		%	\end{align}
	%	where $G_\text{SNR} = (\abs{G_+}^2 + \abs{G_-}^2)/G_\text{noise}$ characterises the ecHT's SNR transfer.
	%	For sufficiently high SNR, the phase standard deviation is approximately
	%	\begin{align}
		%		\sigma_\theta \approx \frac{180}{\pi} \sqrt{\frac{J}{2}} \; \text{[deg]}.
		%	\end{align}
	%	Supplementary Table~\ref{tab:phase_sigma} validates this approximation against Monte Carlo simulations, showing excellent agreement for $\operatorname{SNR}_\text{in} \geq \SI{0}{dB}$. A large-scale noise analysis which holds at low SNR is conducted in Supplementary Section~\ref{sec:noise}.\\
	
	\subsection*{Properties of the calibrated ecHT}
	%These features distinguish the cecHT.
	\begin{enumerate}
		\item \textbf{Universality} The calibration coefficient $C_{\text{opt}}$ is the unique optimal linear estimator regardless of the signal model. For stationary processes beyond single tones, calibration extracts the maximum linearly recoverable signal energy from the ecHT output. We further show that the empirical estimation of $C_\text{opt}$ from $M$ independent windows converges at the rate $1/\sqrt{M}$ (Supplementary Corollary~\ref{thm:asymptotics}), providing a data-driven calibration pathway.
		\item \textbf{Parameter-independence} The ecHT error depends on the length of the window, filter bandwidth and order. After calibration, the bias phase error becomes practically negligible, as seen in Fig.~\ref{fig:simulation}. Therefore, these parameters can be chosen to suit other practical needs.
		\item \textbf{Computational efficiency} Once the calibration parameters $G_+$ and $G_-$ are computed, calibration reduces to multiplication by a precomputed complex number. This adds negligible overhead to the $\mathcal{O}(L \log L)$ FFT-based ecHT, preserving real-time compatibility.
		%\item \textbf{Frequency-specificity} Both $G_+$ and $G_-$ depend on the assumed centre frequency $f_0$. If the oscillation frequency drifts, the deterministic bias $\alpha = \arg G_+$ shifts by approximately $-\Delta\omega \tau_g(\omega_0)$. Consequently, calibration is most effective when $f_0$ is stable or actively tracked. For non-stationary rhythms, periodic re-estimation of $f_0$ and re-computation of $C_{\text{opt}}$ can maintain near-zero mean bias even as the target frequency drifts.
	\end{enumerate}
	
	The combination of analytical transparency, minimal computational cost, and clear design rules makes the cecHT a practical drop-in replacement for existing ecHT pipelines. Our publicly available implementation provides a single Python class that computes $C_\text{opt}$ and applies it automatically, requiring no user intervention beyond a specification of the target frequency and ecHT parameters.% A full discussion of practical recommendations in listed in Supplementary Section~\ref{sec:practical-rec}.
	
	\begin{table}[!ht]
		\centering
		\caption{Phase and amplitude error statistics before and after calibration.}
		\label{tab:schreglmann}
		\begin{tabular}{@{}lSSScSSS@{}}
			\toprule
			\multirow{2.5}{*}{Error} 
			& \multicolumn{3}{c}{Phase [\si{\degree}]} 
			& \multicolumn{3}{c}{Amplitude [\%]} \\ 
			\cmidrule(lr){2-4}\cmidrule(lr){5-7}
			& \multicolumn{1}{c}{Mean} 
			& \multicolumn{1}{c}{Std.Dev.} 
			& \multicolumn{1}{c}{Max} 
			& \multicolumn{1}{c}{Mean} 
			& \multicolumn{1}{c}{Std.Dev.} 
			& \multicolumn{1}{c}{Max} \\ 
			\midrule
			ecHT~\cite{schreglmannNoninvasiveSuppressionEssential2021} 
			& \multicolumn{1}{r}{8.53}  & 1.68 & 11.70 
			& \multicolumn{1}{r}{4.07}  & 2.73 & 8.87 \\
			cecHT 
			& \multicolumn{1}{r}{0.40} & 0.32 & 1.23  
			& \multicolumn{1}{r}{0.70} & 0.55 & 2.20 \\ 
			\bottomrule
		\end{tabular}
	\end{table}

	\subsection*{Simulation}
	
	\begin{figure}[!h]
		\centering
		\includegraphics[width=\columnwidth, trim={0 3em 0 1em}]{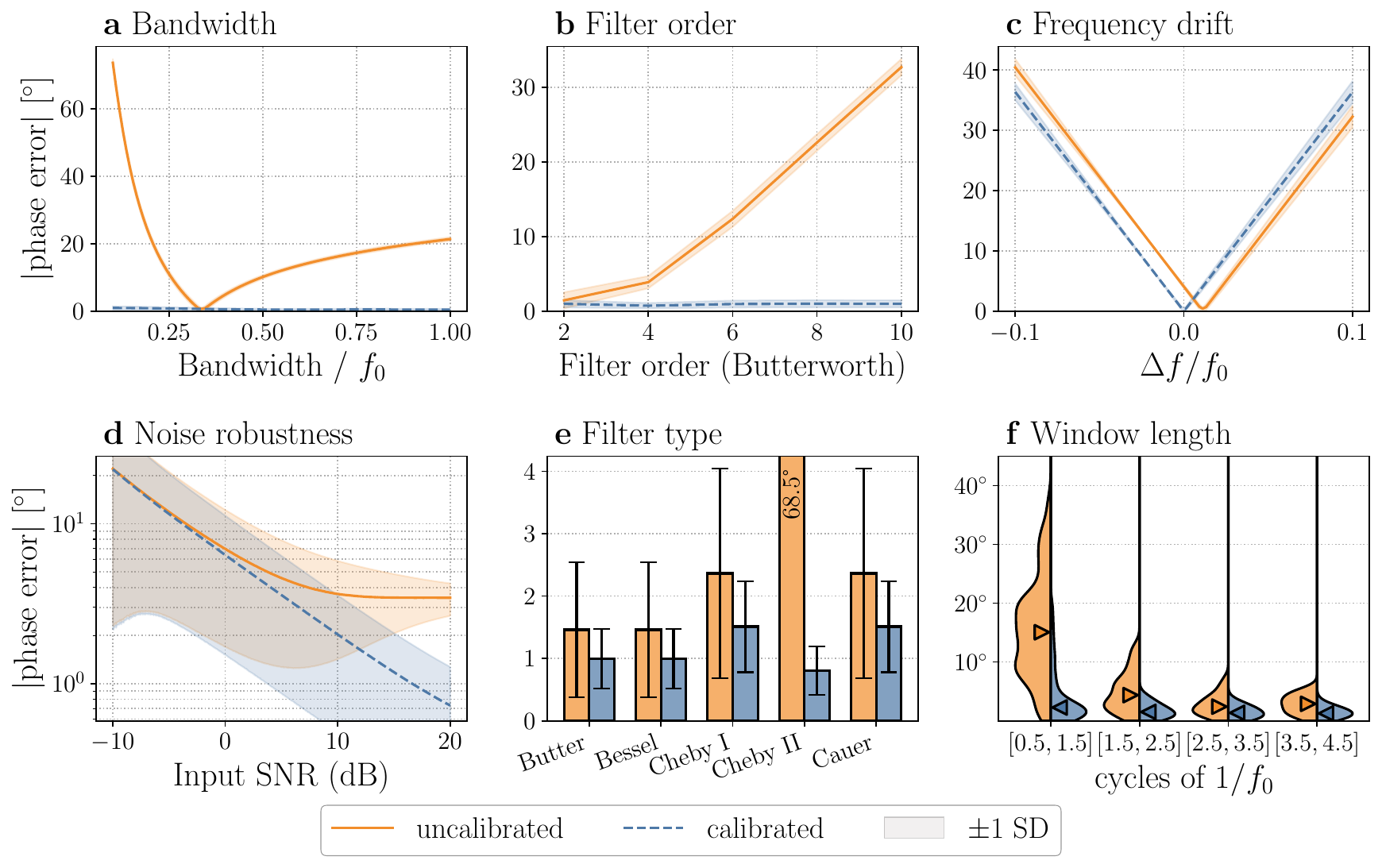}
		\caption{\textbf{Parameter dependence of the ecHT endpoint phase error and effect of scalar calibration in single-tone simulations.} Absolute endpoint phase error (deg) is shown for ecHT (orange) and cecHT (blue), with solid lines giving the mean and shaded bands indicating $\pm 1$ SD. If not noted differently: $f_0=\SI{10}{Hz}$, $F_s=\SI{256}{Hz}$, $N=\operatorname{round}(2.1 F_s/f_0)$, $\text{BW}=[0.7f_0,\, 1.3f_0]$, filter order 2 (lowpass and highpass order 1).~Window length is chosen slightly imperfectly to simulate realistic conditions. Across all panels, the analytic single-tone calibration reduces systematic phase bias while preserving the variability predicted by the complex-gain and leakage analysis. \textbf{a} Dependence on relative filter bandwidth. \textbf{b} Filter order. \textbf{c} Frequency detuning $\Delta f/f_0$. \textbf{d} Input SNR ($\log_{10} y$-axis). \textbf{e} Comparison of IIR filter families (Butterworth, Bessel, Chebyshev I \& II, Cauer/elliptic; bars show mean, error bars SD). \textbf{f} Distribution of absolute endpoint error for non-integer window lengths spanning 1-4 cycles around $f_0$ (half-violins for uncalibrated vs calibrated, markers denote medians).}
		\label{fig:simulation}
	\end{figure}
	 
	Schreglmann et al.~\cite{schreglmannNoninvasiveSuppressionEssential2021} proposed the ecHT, evaluating it through a tone sweep from 2 to \SI{3}{Hz}, using a sampling frequency of $F_s=\SI{256}{Hz}$ and a window length of $N = 256$. The bandpass filter had a bandwidth of \SI{1.5}{Hz} and was centred around \SI{3}{Hz}. The reference phase is defined as the offline phase of the analytic signal~\cite{zrennerShakyGroundTruth2020}. Their reported results can be reproduced and substantially improved through per-frequency calibration, as shown in Table~\ref{tab:schreglmann}. At \SI{2}{Hz} an error of \ang{8.81} would result in a timing offset of \SI{12.2}{ms}.\\
	
	We evaluated the performance of ecHT and cecHT for the case of a single tone at $f_0 = \SI{10}{Hz}$ across varying parameter configurations represented in Fig.~\ref{fig:simulation}. Across all experiments, cecHT consistently outperforms ecHT. Further, Fig.~\ref{fig:simulation} shows how the ecHT, as a parameter-dependent algorithm, becomes virtually independent of parameter choice through calibration. Therefore, any free parameters can be chosen to fit a specific application (e.g.~window length according to frequency stability or computational budget, and filter parameters according to noise estimates or frequency stability). The presented ecHT results are similar to~\cite{schreglmannNoninvasiveSuppressionEssential2021,bresslerWearableEEGSystem2023,liufuOptimizingRealtimePhase2025}. When the ecHT parameters are optimally chosen, the mean error can approach zero without calibration; calibration, however, ensures that this happens consistently.
	
	\subsection*{EEG alpha phase}
	
	\begin{figure}[!ht]
		\centering
		\includegraphics[width=\columnwidth, trim={0 3em 0 1em}]{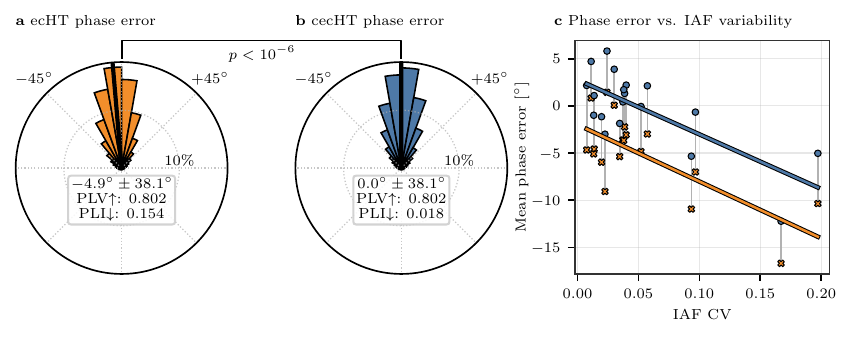}
		\caption{\textbf{Causal alpha-phase estimation with the ecHT and cecHT on the EEG Alpha Waves Dataset.} $F_s=\SI{512}{Hz}$, $N=262$, $\text{BW}=[0.75f_0,\, 1.25f_0]$, filter order 2. \textbf{a} The ecHT shows a systematic phase bias and elevated PLI in the distribution of alpha-band phase errors across all samples. \newline \textbf{b} Applying the analytic calibration removes this bias, preserving PLV, and substantially reducing PLI. \textbf{c} Comparison between per-recording phase error and the coefficient of variation (CV) of the IAF. Across recordings, the cecHT attenuates the dependence of mean phase error on individual alpha-frequency variability, indicating improved robustness to spectral instability.}
		\label{fig:hmc}
	\end{figure}
	
	Data from the \emph{EEG Alpha Waves Dataset} consisting of recordings from 20 participants with eyes-open and -closed segments, was chosen~\cite{cattanEEGAlphaWaves2018}, as the eyes-open/closed paradigm provides a strong modulation of alpha power, offering a reliable benchmark for evaluating whether the phase estimation algorithm correctly tracks the individual alpha frequency (IAF), which was estimated from the Oz-M1 channel using FOOOF~\cite{donoghueParameterizingNeuralPower2020} by fitting an aperiodic ($1/f^\beta$) component and a periodic peak above the noise. Ground truth phase estimates were obtained by applying an offline Hilbert transform to acausally filtered data, using a second-order bandpass filter with cutoffs at 0.75 and 1.25 times the centre frequency. Parameters for the ecHT were chosen according to Bressler et al.~\cite{bresslerWearableEEGSystem2023} and the ecHT error is again comparable to that reported in the literature~\cite{bresslerWearableEEGSystem2023,bresslerRandomizedControlledTrial2024,hebronClosedloopAuditoryStimulation2024}. As Fig.~\ref{fig:hmc} shows, the cecHT reduced the absolute mean error significantly ($p<10^{-6}$) from \ang{4.9} to \ang{0.0}, while retaining the same variance. Significance was tested using a non-parametric circular pairwise permutation test~\cite{goodPermutationParametricBootstrap2005}. The phase-locking value (PLV)~\cite{lachauxMeasuringPhaseSynchrony1999} stayed consistent, but a lower phase-lag index (PLI)~\cite{stamPhaseLagIndex2007} shows a reduction in phase asymmetry.
	
	\subsection*{Tremor phase}
	
	\begin{figure}[!ht]
		\centering
		\includegraphics[width=\columnwidth, trim={0 3em 0 1em}]{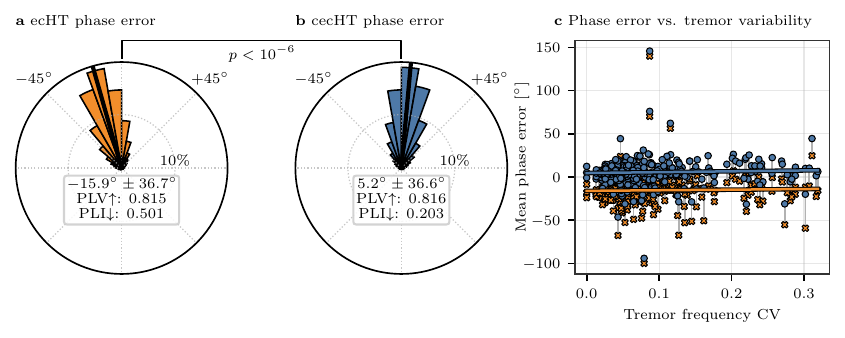}
		\caption{\textbf{Causal phase estimation with the ecHT and cecHT on tremor data.} $F_s\approx \SI{500}{Hz}$, $N=128$, $\text{BW}=[0.5f_0,\, 1.5f_0]$, filter order 4 \textbf{a} The ecHT shows a systematic phase bias and elevated PLI in the distribution of phase errors across all samples. \textbf{b} Applying the analytic calibration removes the bulk of this bias, preserves phase locking (PLV), and substantially reduces PLI. \newline \textbf{c} Comparing per-recording phase error with the coefficient of variation (CV) of the tremor frequency. Across recordings, the cecHT attenuates the dependence of mean phase error on frequency variability, indicating improved robustness to spectral instability.}
		\label{fig:tremor}
	\end{figure}
	
	Schreglmann et al.~\cite{schreglmannNoninvasiveSuppressionEssential2021} investigated how essential tremor (ET) can be reduced by phase-locked transcranial electrical stimulation. Tremor signals are lower in frequency compared to the alpha band. In their experiments, the phase of the ET was calculated using the ecHT, and the patients were then stimulated at pre-defined phases. We estimated the signal phase with the original parameters from~\cite{schreglmannNoninvasiveSuppressionEssential2021}, i.e., a filter bandwidth ranging from $0.5f_0$ to $1.5f_0$, a filter order of 4, and a window length of 128 samples. The calibration results in Fig.~\ref{fig:tremor} show a statistically significant ($p<10^{-6}$) improvement in mean phase error from \ang{-15.9} to \ang{5.2}.\\
	
	Most of the residual post-calibration error stems from an unstable $f_0$. Since $f_0$ can be tracked, calibration can be performed using the most recent estimate. We re-estimated $f_0$ approximately every \SI{4}{s} from the tremor signal using a simple spectral peak estimate, then re-centred the bandpass and re-computed the analytic calibration for the updated $f_0$. This isolates the residual error attributable to leakage/noise from the bias caused by centre-frequency drift, reducing the mean error to \ang{0.0} and slightly lowering the standard deviation, as depicted in Fig.~\ref{fig:tremor_track}. For this application, the system can now be considered fully calibrated.
	
	\begin{figure}[!h]
		\centering
		\includegraphics[width=\columnwidth, trim={0 3em 0 1em}]{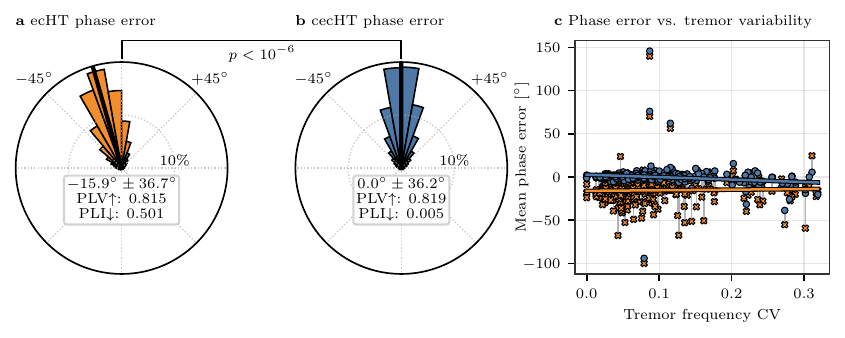}
		\caption{\textbf{Causal tremor-phase estimation with ecHT and cecHT on the tremor data with centre-frequency tracking.} \textbf{a} With centre-frequency tracking but without calibration, ecHT shows a similar behaviour as in Fig.~\ref{fig:tremor}. \newline \textbf{b} With analytic calibration and periodic $f_0$ updates, mean phase bias is reduced to \ang{0.0} while the PLV is preserved and the PLI is strongly reduced. \newline \textbf{c} Mean phase error versus tremor frequency variability, showing improved robustness under drift.}
		\label{fig:tremor_track}
	\end{figure}

	\section*{Discussion}
	
	Reliable phase estimates are essential for closed-loop (brain) stimulation; yet, the state-of-the-art causal filtering and finite-window Hilbert methods are particularly susceptible to errors. The estimator must extrapolate the analytic signal to the boundary of the observation window, where spectral leakage and filter phase distortions are maximal. In this study, we analytically characterised the boundary effects of the endpoint-corrected Hilbert transform and developed a simple calibration that produces a statistically robust estimator with consistent performance and clear design rules.
	
	We showed how the ecHT endpoint for a narrow-band oscillation can be decomposed into two contributions: a desired positive-frequency term $G_+$ and a negative-frequency leakage term $G_-$. This decomposition separates the ecHT errors into a nearly constant phase bias driven by the filter's phase response and a wobble that varies with the unknown instantaneous phase and cannot be eliminated by calibration.
	
	This analysis cleanly separates correctable and intrinsic errors. When negative-frequency leakage is small, the ecHT behaves approximately like the true analytic endpoint, scaled by a complex-valued factor. As leakage increases, the estimator instead yields a mixture, whose phase depends on the (unknown) phase itself. 
	%The deterministic bounds derived in the single-tone setting illustrate this: 
	The phase deviation of the mixture
	%from the positive-frequency component 
	is uniformly bounded by $\arcsin \abs{G_-}/\abs{G_+}$ (c.f.\ Supplementary Theorem~\ref{thm:det-bounds}). In practice, this means that suppressing leakage is not merely \textit{nice to have}; it is essential for keeping the worst-case phase errors small.\\
	
	The scalar calibration factor is the mean-square-optimal linear map from the ecHT endpoint $\hat Z$ to the ideal analytic endpoint $Z$. Geometrically, it is the orthogonal projection of $Z$ onto the subspace spanned by $\hat Z$ and it is statistically equivalent to a Wiener estimator. This has two immediate consequences:
	
	Firstly, calibration targets the systematic gain encoded by $G_+$. This removes the constant phase bias and amplitude scaling when the correlation between $\hat Z$ and $Z$ remains high. Empirically, this is reflected in a substantial reduction in phase errors: systematic phase bias is removed while preserving phase locking and reducing asymmetry. These improvements are significant: a phase bias of only a few degrees can result in a timing offset of several milliseconds, which could be substantial relative to the stimulation window.
	
	Secondly, the Wiener interpretation makes the irreducible error floor explicit. Calibration can only recover information that is present in $\hat Z$. This explains why calibration reduces mean bias more than variability: it fixes the coherent error but not the incoherent component induced by leakage and non-stationarity. In other words, the cecHT provides unbiased endpoint phase estimates under stable conditions, while any remaining jitter reflects either intrinsic limitations of the estimator or genuine instability of the signal.
	
	Stable conditions can be approximated by tracking $f_0$. When the calibration is performed at the (updated) centre frequency, the remaining systematic bias collapses to practically zero, implying that most of the residual mean error was drift-induced rather than a failure of the calibration.\\
	
	%Phase estimators should be evaluated by both their absolute error and the structure of their residuals. Does the ecHT primarily introduce a constant phase offset or phase-dependent distortions? The ecHT can introduce a systematic bias, meaning that estimator artefacts may appear as physiological phase asymmetries. Calibration substantially reduces this, while maintaining PLV, thereby improving the interpretability of downstream statistics.
	
	The cecHT remains computationally efficient, as it is FFT-based with $\mathcal{O}(L\log L)$ complexity. The additional endpoint correction and calibration produce only modest overhead. This keeps the method compatible with embedded or real-time pipelines, where latency budgets are tight.\ The calibration can be used as a drop-in improvement for existing ecHT pipelines.\\
	
	%The calibration factor is frequency-specific because both $G_+$ and $G_-$ both depend on the way in which the filter samples the spectrum around the assumed centre frequency $f_0$. If $f_0$ drifts, the phase bias becomes time-varying, which can introduce bias if the calibration is fixed. This motivates adaptive frequency tracking and periodic recalibration.
	
	The ecHT assumes stationarity over the analysis window. Transient bursts, artefacts, and within-window changes violate this assumption and can amplify leakage. Robust estimation or explicit artefact rejection is a necessary companion for real-world deployment. Also, phase may be poorly defined during low-amplitude windows.
	
	When multiple nearby frequency components coexist, a single centre frequency does not fully describe the signal. In such settings, the scalar calibration remains the MSE-optimal correction for the reference definition of $Z$, but its meaning depends on how the reference signal is constructed. More importantly, true phase can become model-dependent.\ Performance claims should be interpreted as agreement with a specified reference pipeline, not an absolute ground truth.\\
	
	The theoretical characterisation here suggests several concrete extensions that could further enhance cecHT for real-time (neural) signal processing.
	
	Optimal calibration may be defined not only for the single-tone case, but also for general complex random variables. This would enable the ecHT to be calibrated against a reference constructed from real data, providing an opportunity to learn the most effective correction for the actual environment. The empirical estimator's asymptotic properties provide a basis for estimating uncertainty in calibration and for deciding when recalibration is statistically warranted.
	
	Applications may require the simultaneous tracking of multiple rhythms or the distinction of several narrow-band oscillations. A cecHT filter bank using one analytic mask per band and a band-specific calibration would allow phase estimation across all relevant frequencies, while making inter-band interference and leakage measurable.
	
	Finally, a multivariate generalisation would replace the scalar coefficient $C_\text{opt}$ by a matrix that maps a vector of ecHT endpoints (across channels or spatially filtered components) onto a reference vector.\ This could unify endpoint correction with spatial filtering, promising a significant impact on EEG/MEG and multi-contact invasive recordings in particular.\\
	
	These results establish the cecHT as an intermediate solution between simple causal filters and more complex adaptive models. By decomposing endpoint distortion into a correctable gain term and an irreducible leakage term, our framework offers both a practical tool and a diagnostic language to explain when and why phase estimation fails. The combination of theoretical transparency, computational efficiency, and a clear path to future extensions makes the ecHT a strong candidate for standardised real-time phase estimation, provided that its requirements regarding frequency stability and artefact robustness are addressed.
	
	\section*{Methods}
	
	\subsection*{Analytical characterisation of the ecHT error}
	
	For a narrow-band oscillation approximated using 
	\begin{align}
		x(n) =\cos(\omega_0 n + \varphi_0)
		\label{eq:x(n)}
	\end{align}
	at an estimated frequency $\omega_0=2\pi f_0/F_s$ with unknown initial phase $\varphi_0$, we derived an explicit expression for the ecHT endpoint output. We show in Supplementary Section~\ref{sec:DFT}, how the effective complex gain at the endpoint $F$ decomposes into two contributions
	\begin{align}
		F(\varphi_0) = G_+ + G_- \e^{-\j2\varphi_0}.
	\end{align}
	$G_+$ is a deterministic complex gain acting on the positive-frequency component which reflects the filter's phase response and is constant across cycles. The second term quantifies residual negative-frequency leakage that escapes the analytic mask and filter, and produces cycle-to-cycle variability.\\
	
	The total phase error can be written as 
	\begin{align}
		\varepsilon_\theta(\varphi_0) = \alpha + \delta_\theta(\varphi_0),
	\end{align} 
	where $\alpha = \arg G_+$ is the calibratable phase bias and $\delta_\theta(\varphi_0)$ is the leakage-induced ripple. In Supplementary Theorem~\ref{thm:det-bounds}, we show that for all initial phases, the ripple satisfies a uniform bound
	\begin{align}
		\abs{\delta_\theta(\varphi_0)} \leq \arcsin r, \quad r = \frac{\abs{G_-}}{\abs{G_+}},
	\end{align}
	where $r$ is the leakage ratio. Consequently, if the oscillation is stable and the filter successfully suppresses most of the negative-frequency components (e.g.~$r=0.1$), assuming Eq.~\eqref{eq:x(n)} holds after calibrating $\alpha$ at the newest sample, the worst-case residual phase error is $\arcsin(0.1) \approx \ang{5.7}$. In practice, typical leakage ratios are maximally on the order of a few percent, yielding only a few degrees of worst-case bias.\\
	
	Similarly, by Theorem~\ref{thm:det-bounds}, the amplitude error satisfies
	\begin{align}
		\abs{\varepsilon_A(\varphi_0)} = \big|\abs{F(\varphi_0)}-1\big|
		\leq \big|\abs{G_+}-1\big| + \abs{G_-},
	\end{align}
	demonstrating that amplitude distortion is also bounded by the deterministic gain mismatch and leakage power.\\
	
	The frequency sensitivity of the deterministic bias $\alpha = \arg G_+$ defined in Supplementary Section~\ref{sec:bias-ripple-dec} represents the endpoint group delay
	\begin{align}
		\tau_g(\omega)=-\dv{\omega} \arg G_+(\omega)
	\end{align} 
	of the ecHT operator. If the true oscillation frequency drifts by $\Delta\omega$ from the assumed centre frequency, the phase bias changes approximately as
	\begin{align}
		\alpha(\omega_0 + \Delta\omega) \approx \alpha(\omega_0) - \Delta\omega \, \tau_g(\omega_0).
	\end{align}
	This relationship reveals a fundamental trade-off: narrow-band filters reduce leakage (smaller $\abs{G_-}$) but increase group delay $\tau_g$, making the phase estimate more sensitive to frequency mismatch.

	\section*{Data availability}
	
	All data used in this study is freely available. The EEG Alpha Wave Dataset~\cite{cattanEEGAlphaWaves2018} is available from MOABB~\cite{aristimunhaMotherAllBCI2025}. The tremor data of Schreglmann et al.~\cite{schreglmannNoninvasiveSuppressionEssential2021} is available at the Harvard Dataverse~\cite{schreglmannReplicationDataNoninvasive2020}.
	
	\section*{Code availability}
	
	Code for the performed experiments and the cecHT algorithm is available from GitHub at \url{https://github.com/eikeosmers/cecHT}
	
	\printbibliography

@software{aristimunhaMotherAllBCI2025,
  title = {Mother of All {{BCI Benchmarks}}},
  author = {Aristimunha, Bruno and Carrara, Igor and Guetschel, Pierre and Sedlar, Sara and Rodrigues, Pedro and Sosulski, Jan and Narayanan, Divyesh and Bjareholt, Erik and Barthelemy, Quentin and Schirrmeister, Robin Tibor and Kobler, Reinmar and Kalunga, Emmanuel and Darmet, Ludovic and Gregoire, Cattan and Abdul Hussain, Ali and Gatti, Ramiro and Goncharenko, Vladislav and Thielen, Jordy and Moreau, Thomas and Roy, Yannick and Jayaram, Vinay and Barachant, Alexandre and Chevallier, Sylvain},
  date = {2025-11-08},
  doi = {10.5281/ZENODO.10034223},
  url = {https://zenodo.org/doi/10.5281/zenodo.10034223},
  urldate = {2026-02-16},
  organization = {Zenodo},
  version = {v1.4.0}
}

@article{blancoStationarityEEGSeries1995,
  title = {Stationarity of the {{EEG}} Series},
  author = {Blanco, S. and Garcia, H. and Quiroga, R.Q. and Romanelli, L. and Rosso, O.A.},
  date = {1995-07},
  journaltitle = {IEEE Engineering in Medicine and Biology Magazine},
  shortjournal = {IEEE Eng. Med. Biol. Mag.},
  volume = {14},
  number = {4},
  pages = {395--399},
  issn = {07395175},
  doi = {10.1109/51.395321},
  url = {http://ieeexplore.ieee.org/document/395321/},
  urldate = {2025-12-02}
}

@article{bresslerRandomizedControlledTrial2024,
  title = {A Randomized Controlled Trial of Alpha Phase-Locked Auditory Stimulation to Treat Symptoms of Sleep Onset Insomnia},
  author = {Bressler, Scott and Neely, Ryan and Yost, Ryan M and Wang, David},
  date = {2024-06-06},
  journaltitle = {Scientific Reports},
  shortjournal = {Sci Rep},
  volume = {14},
  number = {1},
  pages = {13039},
  issn = {2045-2322},
  doi = {10.1038/s41598-024-63385-1},
  url = {https://www.nature.com/articles/s41598-024-63385-1},
  urldate = {2024-07-11},
  langid = {english}
}

@article{bresslerWearableEEGSystem2023,
  title = {A Wearable {{EEG}} System for Closed-Loop Neuromodulation of Sleep-Related Oscillations},
  author = {Bressler, Scott and Neely, Ryan and Yost, Ryan M and Wang, David and Read, Heather L},
  date = {2023-10-01},
  journaltitle = {Journal of Neural Engineering},
  shortjournal = {J. Neural Eng.},
  volume = {20},
  number = {5},
  pages = {056030},
  issn = {1741-2560, 1741-2552},
  doi = {10.1088/1741-2552/acfb3b},
  url = {https://iopscience.iop.org/article/10.1088/1741-2552/acfb3b},
  urldate = {2024-07-12}
}

@article{buschRealtimePhaseAmplitude2022,
  title = {Real-Time Phase and Amplitude Estimation of Neurophysiological Signals Exploiting a Non-Resonant Oscillator},
  author = {Busch, Johannes L. and Feldmann, Lucia K. and Kühn, Andrea A. and Rosenblum, Michael},
  date = {2022-01},
  journaltitle = {Experimental Neurology},
  shortjournal = {Experimental Neurology},
  volume = {347},
  pages = {113869},
  issn = {00144886},
  doi = {10.1016/j.expneurol.2021.113869},
  url = {https://linkinghub.elsevier.com/retrieve/pii/S0014488621002776},
  urldate = {2024-09-24},
  langid = {english}
}

@book{casellaStatisticalInference2002,
  title = {Statistical Inference},
  author = {Casella, George and Berger, Roger L.},
  date = {2002},
  edition = {2. ed},
  publisher = {Duxbury},
  location = {Pacific Grove, Calif},
  isbn = {978-0-534-24312-8},
  langid = {english},
  pagetotal = {660}
}

@dataset{cattanEEGAlphaWaves2018,
  title = {{{EEG Alpha Waves}} Dataset},
  author = {Cattan, Grégoire and Rodrigues, Pedro L. C. and Congedo, Marco},
  date = {2018-12-17},
  publisher = {Zenodo},
  doi = {10.5281/ZENODO.2348891},
  url = {https://zenodo.org/record/2348891},
  urldate = {2026-02-16}
}

@article{choDevelopmentEvaluationRealTime2024,
  title = {Development and {{Evaluation}} of a {{Real-Time Phase-Triggered Stimulation Algorithm}} for the {{CorTec Brain Interchange}}},
  author = {Cho, Hanbin and Benjaber, Moaad and Alexis Gkogkidis, C. and Buchheit, Marina and Ruiz-Rodríguez, Juan F. and Grannan, Benjamin L. and Weaver, Kurt E. and Ko, Andrew L. and Cramer, Steven C. and Ojemann, Jeffrey G. and Denison, Timothy and Herron, Jeffrey A.},
  date = {2024},
  journaltitle = {IEEE Transactions on Neural Systems and Rehabilitation Engineering},
  shortjournal = {IEEE Trans. Neural Syst. Rehabil. Eng.},
  volume = {32},
  pages = {3625--3635},
  issn = {1534-4320, 1558-0210},
  doi = {10.1109/TNSRE.2024.3459801},
  url = {https://ieeexplore.ieee.org/document/10679661/},
  urldate = {2025-12-10}
}

@article{corcoranReliableAutomatedMethod2018,
  title = {Toward a Reliable, Automated Method of Individual Alpha Frequency ({{IAF}}) Quantification},
  author = {Corcoran, Andrew W. and Alday, Phillip M. and Schlesewsky, Matthias and Bornkessel‐Schlesewsky, Ina},
  date = {2018-07},
  journaltitle = {Psychophysiology},
  shortjournal = {Psychophysiology},
  volume = {55},
  number = {7},
  pages = {e13064},
  issn = {0048-5772, 1469-8986},
  doi = {10.1111/psyp.13064},
  url = {https://onlinelibrary.wiley.com/doi/10.1111/psyp.13064},
  urldate = {2025-12-02},
  langid = {english}
}

@article{donoghueParameterizingNeuralPower2020,
  title = {Parameterizing Neural Power Spectra into Periodic and Aperiodic Components},
  author = {Donoghue, Thomas and Haller, Matar and Peterson, Erik J. and Varma, Paroma and Sebastian, Priyadarshini and Gao, Richard and Noto, Torben and Lara, Antonio H. and Wallis, Joni D. and Knight, Robert T. and Shestyuk, Avgusta and Voytek, Bradley},
  date = {2020-12},
  journaltitle = {Nature Neuroscience},
  shortjournal = {Nat Neurosci},
  volume = {23},
  number = {12},
  pages = {1655--1665},
  issn = {1097-6256, 1546-1726},
  doi = {10.1038/s41593-020-00744-x},
  url = {https://www.nature.com/articles/s41593-020-00744-x},
  urldate = {2026-01-27},
  langid = {english}
}

@article{fersterBenchmarkingRealTimeAlgorithms2022,
  title = {Benchmarking {{Real-Time Algorithms}} for {{In-Phase Auditory Stimulation}} of {{Low Amplitude Slow Waves With Wearable EEG Devices During Sleep}}},
  author = {Ferster, Maria Laura and Da Poian, Giulia and Menachery, Kiran and Schreiner, Simon J. and Lustenberger, Caroline and Maric, Angelina and Huber, Reto and Baumann, Christian R. and Karlen, Walter},
  date = {2022-09},
  journaltitle = {IEEE Transactions on Biomedical Engineering},
  shortjournal = {IEEE Trans. Biomed. Eng.},
  volume = {69},
  number = {9},
  pages = {2916--2925},
  issn = {0018-9294, 1558-2531},
  doi = {10.1109/TBME.2022.3157468},
  url = {https://ieeexplore.ieee.org/document/9730073/},
  urldate = {2024-09-19}
}

@article{fingelkurtsEditorialEEGPhenomenology2010,
  title = {Editorial: {{EEG Phenomenology}} and {{Multiple Faces}} of {{Short-term EEG Spectral Pattern}}},
  shorttitle = {Editorial},
  author = {Fingelkurts, Al. A and Fingelkurts, An. A},
  date = {2010-09-08},
  journaltitle = {The Open Neuroimaging Journal},
  shortjournal = {TONIJ},
  volume = {4},
  number = {1},
  pages = {111--113},
  issn = {1874-4400},
  doi = {10.2174/1874440001004010111},
  url = {https://openneuroimagingjournal.com/VOLUME/04/PAGE/111/},
  urldate = {2025-12-02},
  langid = {english}
}

@article{florianDynamicSpectralAnalysis1995,
  title = {Dynamic Spectral Analysis of Event-Related {{EEG}} Data},
  author = {Florian, Gernot and Pfurtscheller, Gert},
  date = {1995-11},
  journaltitle = {Electroencephalography and Clinical Neurophysiology},
  shortjournal = {Electroencephalography and Clinical Neurophysiology},
  volume = {95},
  number = {5},
  pages = {393--396},
  issn = {00134694},
  doi = {10.1016/0013-4694(95)00198-8},
  url = {https://linkinghub.elsevier.com/retrieve/pii/0013469495001988},
  urldate = {2025-12-02},
  langid = {english}
}

@book{goodPermutationParametricBootstrap2005,
  title = {Permutation, {{Parametric}} and {{Bootstrap Tests}} of {{Hypotheses}}},
  author = {Good, Phillip},
  date = {2005},
  series = {Springer {{Series}} in {{Statistics}}},
  publisher = {Springer-Verlag},
  location = {New York},
  doi = {10.1007/b138696},
  url = {http://link.springer.com/10.1007/b138696},
  urldate = {2026-02-16},
  isbn = {978-0-387-20279-2},
  langid = {english}
}

@article{grossAnalyticalMethodsExperimental2014,
  title = {Analytical Methods and Experimental Approaches for Electrophysiological Studies of Brain Oscillations},
  author = {Gross, Joachim},
  date = {2014-05},
  journaltitle = {Journal of Neuroscience Methods},
  shortjournal = {Journal of Neuroscience Methods},
  volume = {228},
  pages = {57--66},
  issn = {01650270},
  doi = {10.1016/j.jneumeth.2014.03.007},
  url = {https://linkinghub.elsevier.com/retrieve/pii/S0165027014000958},
  urldate = {2025-12-10},
  langid = {english}
}

@article{harlowIndividualizedClosedLoopAcoustic2024,
  title = {Individualized {{Closed-Loop Acoustic Stimulation Suggests}} an {{Alpha Phase Dependence}} of {{Sound Evoked}} and {{Induced Brain Activity Measured}} with {{EEG Recordings}}},
  author = {Harlow, Tylor J. and Marquez, Samantha M. and Bressler, Scott and Read, Heather L.},
  date = {2024-06},
  journaltitle = {eNeuro},
  shortjournal = {eNeuro},
  volume = {11},
  number = {6},
  pages = {ENEURO.0511-23.2024},
  issn = {2373-2822},
  doi = {10.1523/ENEURO.0511-23.2024},
  url = {https://www.eneuro.org/lookup/doi/10.1523/ENEURO.0511-23.2024},
  urldate = {2025-11-06},
  langid = {english}
}

@article{hebronClosedloopAuditoryStimulation2024,
  title = {A Closed-Loop Auditory Stimulation Approach Selectively Modulates Alpha Oscillations and Sleep Onset Dynamics in Humans},
  author = {Hebron, Henry and Lugli, Beatrice and Dimitrova, Radost and Jaramillo, Valeria and Yeh, Lisa R. and Rhodes, Edward and Grossman, Nir and Dijk, Derk-Jan and Violante, Ines R.},
  editor = {Hanslmayr, Simon},
  date = {2024-06-18},
  journaltitle = {PLOS Biology},
  shortjournal = {PLoS Biol},
  volume = {22},
  number = {6},
  pages = {e3002651},
  issn = {1545-7885},
  doi = {10.1371/journal.pbio.3002651},
  url = {https://dx.plos.org/10.1371/journal.pbio.3002651},
  urldate = {2025-03-20},
  langid = {english}
}

@article{jaramilloClosedloopAuditoryStimulation2024,
  title = {Closed-Loop Auditory Stimulation Targeting Alpha and Theta Oscillations during Rapid Eye Movement Sleep Induces Phase-Dependent Power and Frequency Changes},
  author = {Jaramillo, Valeria and Hebron, Henry and Wong, Sara and Atzori, Giuseppe and Bartsch, Ullrich and Dijk, Derk-Jan and Violante, Ines R},
  date = {2024-12-11},
  journaltitle = {SLEEP},
  volume = {47},
  number = {12},
  pages = {zsae193},
  issn = {0161-8105, 1550-9109},
  doi = {10.1093/sleep/zsae193},
  url = {https://academic.oup.com/sleep/article/doi/10.1093/sleep/zsae193/7745355},
  urldate = {2025-11-06},
  langid = {english}
}

@article{kimEEGPhaseCan2023,
  title = {{{EEG Phase Can Be Predicted}} with {{Similar Accuracy}} across {{Cognitive States}} after {{Accounting}} for {{Power}} and {{Signal-to-Noise Ratio}}},
  author = {Kim, Brian and Erickson, Brian A. and Fernandez-Nunez, Guadalupe and Rich, Ryan and Mentzelopoulos, Georgios and Vitale, Flavia and Medaglia, John D.},
  date = {2023-09},
  journaltitle = {eNeuro},
  shortjournal = {eNeuro},
  volume = {10},
  number = {9},
  pages = {ENEURO.0050-23.2023},
  issn = {2373-2822},
  doi = {10.1523/ENEURO.0050-23.2023},
  url = {https://www.eneuro.org/lookup/doi/10.1523/ENEURO.0050-23.2023},
  urldate = {2025-12-09},
  langid = {english}
}

@article{lachauxMeasuringPhaseSynchrony1999,
  title = {Measuring Phase Synchrony in Brain Signals},
  author = {Lachaux, Jean-Philippe and Rodriguez, Eugenio and Martinerie, Jacques and Varela, Francisco J.},
  date = {1999},
  journaltitle = {Human Brain Mapping},
  shortjournal = {Hum. Brain Mapp.},
  volume = {8},
  number = {4},
  pages = {194--208},
  issn = {1065-9471, 1097-0193},
  doi = {10.1002/(sici)1097-0193(1999)8:4%3C194::aid-hbm4%3E3.0.co;2-c},
  url = {https://onlinelibrary.wiley.com/doi/10.1002/(SICI)1097-0193(1999)8:4<194::AID-HBM4>3.0.CO;2-C},
  urldate = {2024-10-15},
  langid = {english}
}

@article{lawhernDetectingAlphaSpindle2013,
  title = {Detecting Alpha Spindle Events in {{EEG}} Time Series Using Adaptive Autoregressive Models},
  author = {Lawhern, Vernon and Kerick, Scott and Robbins, Kay A},
  date = {2013-12},
  journaltitle = {BMC Neuroscience},
  shortjournal = {BMC Neurosci},
  volume = {14},
  number = {1},
  pages = {101},
  issn = {1471-2202},
  doi = {10.1186/1471-2202-14-101},
  url = {https://bmcneurosci.biomedcentral.com/articles/10.1186/1471-2202-14-101},
  urldate = {2025-12-02},
  langid = {english}
}

@article{liufuOptimizingRealtimePhase2025,
  title = {Optimizing Real-Time Phase Detection in Diverse Rhythmic Biological Signals for Phase-Specific Neurostimulation},
  author = {Liufu, Mengzhan and Leveroni, Zachary M and Shridhar, Sameera and Zhou, Nan and Yu, Jai Y},
  date = {2025-10-01},
  journaltitle = {Journal of Neural Engineering},
  shortjournal = {J. Neural Eng.},
  volume = {22},
  number = {5},
  pages = {056037},
  issn = {1741-2560, 1741-2552},
  doi = {10.1088/1741-2552/ae10e1},
  url = {https://iopscience.iop.org/article/10.1088/1741-2552/ae10e1},
  urldate = {2025-11-26}
}

@article{liuNovelRealtimePhase2025,
  title = {A {{Novel Real-time Phase Prediction Network}} in {{EEG Rhythm}}},
  author = {Liu, Hao and Qi, Zihui and Wang, Yihang and Yang, Zhengyi and Fan, Lingzhong and Zuo, Nianming and Jiang, Tianzi},
  date = {2025-03},
  journaltitle = {Neuroscience Bulletin},
  shortjournal = {Neurosci. Bull.},
  volume = {41},
  number = {3},
  pages = {391--405},
  issn = {1673-7067, 1995-8218},
  doi = {10.1007/s12264-024-01321-z},
  url = {https://link.springer.com/10.1007/s12264-024-01321-z},
  urldate = {2025-12-10},
  langid = {english}
}

@article{madsenNoTracePhase2019,
  title = {No Trace of Phase: {{Corticomotor}} Excitability Is Not Tuned by Phase of Pericentral Mu-Rhythm},
  shorttitle = {No Trace of Phase},
  author = {Madsen, Kristoffer Hougaard and Karabanov, Anke Ninija and Krohne, Lærke Gebser and Safeldt, Mads Gylling and Tomasevic, Leo and Siebner, Hartwig Roman},
  date = {2019-09},
  journaltitle = {Brain Stimulation},
  shortjournal = {Brain Stimulation},
  volume = {12},
  number = {5},
  pages = {1261--1270},
  issn = {1935861X},
  doi = {10.1016/j.brs.2019.05.005},
  url = {https://linkinghub.elsevier.com/retrieve/pii/S1935861X19302128},
  urldate = {2025-12-09},
  langid = {english}
}

@article{makarovaHardwareenabledLowLatency2025,
  title = {Hardware-Enabled Low Latency Rhythmic Brain State Tracking for Brain Stimulation Applications},
  author = {Makarova, Milana and Fedosov, Nikita and Nekrasova, Julia and Ossadtchi, Alexey},
  date = {2025-10},
  journaltitle = {NeuroImage},
  shortjournal = {NeuroImage},
  volume = {319},
  pages = {121437},
  issn = {10538119},
  doi = {10.1016/j.neuroimage.2025.121437},
  url = {https://linkinghub.elsevier.com/retrieve/pii/S1053811925004409},
  urldate = {2025-12-10},
  langid = {english}
}

@article{mcintoshEstimationPhaseEEG2020,
  title = {Estimation of Phase in {{EEG}} Rhythms for Real-Time Applications},
  author = {McIntosh, J R and Sajda, P},
  date = {2020-06-02},
  journaltitle = {Journal of Neural Engineering},
  shortjournal = {J. Neural Eng.},
  volume = {17},
  number = {3},
  pages = {034002},
  issn = {1741-2552},
  doi = {10.1088/1741-2552/ab8683},
  url = {https://iopscience.iop.org/article/10.1088/1741-2552/ab8683},
  urldate = {2024-10-22}
}

@article{onojimaStateinformedStimulationApproach2021,
  title = {A State-Informed Stimulation Approach with Real-Time Estimation of the Instantaneous Phase of Neural Oscillations by a {{Kalman}} Filter},
  author = {Onojima, Takayuki and Kitajo, Keiichi},
  date = {2021-12-01},
  journaltitle = {Journal of Neural Engineering},
  shortjournal = {J. Neural Eng.},
  volume = {18},
  number = {6},
  pages = {066001},
  issn = {1741-2560, 1741-2552},
  doi = {10.1088/1741-2552/ac2f7b},
  url = {https://iopscience.iop.org/article/10.1088/1741-2552/ac2f7b},
  urldate = {2025-12-10}
}

@book{oppenheimDiscretetimeSignalProcessing2014,
  title = {Discrete-Time Signal Processing},
  author = {Oppenheim, Alan V. and Schafer, Ronald W.},
  date = {2014},
  edition = {Third international edition},
  publisher = {Pearson, Dorling Kindersley},
  location = {Noida (India)},
  isbn = {978-93-325-3503-9},
  langid = {english}
}

@article{rosenblumRealtimeEstimationPhase2021,
  title = {Real-Time Estimation of Phase and Amplitude with Application to Neural Data},
  author = {Rosenblum, Michael and Pikovsky, Arkady and Kühn, Andrea A. and Busch, Johannes L.},
  date = {2021-09-10},
  journaltitle = {Scientific Reports},
  shortjournal = {Sci Rep},
  volume = {11},
  number = {1},
  pages = {18037},
  issn = {2045-2322},
  doi = {10.1038/s41598-021-97560-5},
  url = {https://www.nature.com/articles/s41598-021-97560-5},
  urldate = {2024-09-24},
  langid = {english}
}

@article{santostasiPhaselockedLoopPrecisely2016,
  title = {Phase-Locked Loop for Precisely Timed Acoustic Stimulation during Sleep},
  author = {Santostasi, Giovanni and Malkani, Roneil and Riedner, Brady and Bellesi, Michele and Tononi, Giulio and Paller, Ken A. and Zee, Phyllis C.},
  date = {2016-02},
  journaltitle = {Journal of Neuroscience Methods},
  shortjournal = {Journal of Neuroscience Methods},
  volume = {259},
  pages = {101--114},
  issn = {01650270},
  doi = {10.1016/j.jneumeth.2015.11.007},
  url = {https://linkinghub.elsevier.com/retrieve/pii/S016502701500401X},
  urldate = {2025-02-05},
  langid = {english}
}

@article{schreglmannNoninvasiveSuppressionEssential2021,
  title = {Non-Invasive Suppression of Essential Tremor via Phase-Locked Disruption of Its Temporal Coherence},
  author = {Schreglmann, Sebastian R. and Wang, David and Peach, Robert L. and Li, Junheng and Zhang, Xu and Latorre, Anna and Rhodes, Edward and Panella, Emanuele and Cassara, Antonino M. and Boyden, Edward S. and Barahona, Mauricio and Santaniello, Sabato and Rothwell, John and Bhatia, Kailash P. and Grossman, Nir},
  date = {2021-01-13},
  journaltitle = {Nature Communications},
  shortjournal = {Nat Commun},
  volume = {12},
  number = {1},
  pages = {363},
  issn = {2041-1723},
  doi = {10.1038/s41467-020-20581-7},
  url = {https://www.nature.com/articles/s41467-020-20581-7},
  urldate = {2024-07-18},
  langid = {english}
}

@dataset{schreglmannReplicationDataNoninvasive2020,
  title = {Replication {{Data}} for: {{Non-invasive Suppression}} of {{Essential Tremor}} via {{Phase-Locked Disruption}} of Its {{Temporal Coherence}}},
  shorttitle = {Replication {{Data}} For},
  author = {Schreglmann, Sebastian and Peach, Robert L. and Li, Junheng and Rhodes, Edward and Latorre, Anna and Rothwell, John and Bhatia, Kailash P. and Grossman, Nir},
  namea = {Peach, Robert},
  nameatype = {collaborator},
  date = {2020},
  pages = {216134158, 314469561, 500861282, 3509},
  publisher = {Harvard Dataverse},
  doi = {10.7910/DVN/Z6EN2I},
  url = {https://dataverse.harvard.edu/citation?persistentId=doi:10.7910/DVN/Z6EN2I},
  urldate = {2026-01-20},
  version = {1.0}
}

@article{shirinpourExperimentalEvaluationMethods2020,
  title = {Experimental Evaluation of Methods for Real-Time {{EEG}} Phase-Specific Transcranial Magnetic Stimulation},
  author = {Shirinpour, Sina and Alekseichuk, Ivan and Mantell, Kathleen and Opitz, Alexander},
  date = {2020-07-13},
  journaltitle = {Journal of Neural Engineering},
  shortjournal = {J. Neural Eng.},
  volume = {17},
  number = {4},
  pages = {046002},
  issn = {1741-2552},
  doi = {10.1088/1741-2552/ab9dba},
  url = {https://iopscience.iop.org/article/10.1088/1741-2552/ab9dba},
  urldate = {2024-07-18}
}

@article{stamPhaseLagIndex2007,
  title = {Phase Lag Index: {{Assessment}} of Functional Connectivity from Multi Channel {{EEG}} and {{MEG}} with Diminished Bias from Common Sources},
  shorttitle = {Phase Lag Index},
  author = {Stam, Cornelis J. and Nolte, Guido and Daffertshofer, Andreas},
  date = {2007-11},
  journaltitle = {Human Brain Mapping},
  shortjournal = {Human Brain Mapping},
  volume = {28},
  number = {11},
  pages = {1178--1193},
  issn = {1065-9471, 1097-0193},
  doi = {10.1002/hbm.20346},
  url = {https://onlinelibrary.wiley.com/doi/10.1002/hbm.20346},
  urldate = {2024-10-15},
  langid = {english}
}

@article{vinao-carlJustPhaseCausal2024,
  title = {Just a Phase? {{Causal}} Probing Reveals Spurious Phasic Dependence of Sustained Attention},
  shorttitle = {Just a Phase?},
  author = {Vinao-Carl, M. and Gal-Shohet, Y. and Rhodes, E. and Li, J. and Hampshire, A. and Sharp, D. and Grossman, N.},
  date = {2024-01},
  journaltitle = {NeuroImage},
  shortjournal = {NeuroImage},
  volume = {285},
  pages = {120477},
  issn = {10538119},
  doi = {10.1016/j.neuroimage.2023.120477},
  url = {https://linkinghub.elsevier.com/retrieve/pii/S1053811923006274},
  urldate = {2025-03-19},
  langid = {english}
}

@article{wodeyarDifferentMethodsEstimate2023,
  title = {Different {{Methods}} to {{Estimate}} the {{Phase}} of {{Neural Rhythms Agree But Only During Times}} of {{Low Uncertainty}}},
  author = {Wodeyar, Anirudh and Marshall, François A. and Chu, Catherine J. and Eden, Uri T. and Kramer, Mark A.},
  date = {2023-11},
  journaltitle = {eNeuro},
  shortjournal = {eNeuro},
  volume = {10},
  number = {11},
  pages = {ENEURO.0507-22.2023},
  issn = {2373-2822},
  doi = {10.1523/ENEURO.0507-22.2023},
  url = {https://www.eneuro.org/lookup/doi/10.1523/ENEURO.0507-22.2023},
  urldate = {2025-12-10},
  langid = {english}
}

@article{wodeyarStateSpaceModeling2021,
  title = {A State Space Modeling Approach to Real-Time Phase Estimation},
  author = {Wodeyar, Anirudh and Schatza, Mark and Widge, Alik S and Eden, Uri T and Kramer, Mark A},
  date = {2021-09-27},
  journaltitle = {eLife},
  volume = {10},
  pages = {e68803},
  issn = {2050-084X},
  doi = {10.7554/eLife.68803},
  url = {https://elifesciences.org/articles/68803},
  urldate = {2025-12-10},
  langid = {english}
}

@article{zrennerClosedLoopBrainStimulation2024,
  title = {Closed-{{Loop Brain Stimulation}}},
  author = {Zrenner, Christoph and Ziemann, Ulf},
  date = {2024-03},
  journaltitle = {Biological Psychiatry},
  shortjournal = {Biological Psychiatry},
  volume = {95},
  number = {6},
  pages = {545--552},
  issn = {00063223},
  doi = {10.1016/j.biopsych.2023.09.014},
  url = {https://linkinghub.elsevier.com/retrieve/pii/S0006322323015895},
  urldate = {2025-02-10},
  langid = {english}
}

@article{zrennerRealtimeEEGdefinedExcitability2018,
  title = {Real-Time {{EEG-defined}} Excitability States Determine Efficacy of {{TMS-induced}} Plasticity in Human Motor Cortex},
  author = {Zrenner, Christoph and Desideri, Debora and Belardinelli, Paolo and Ziemann, Ulf},
  date = {2018-03},
  journaltitle = {Brain Stimulation},
  shortjournal = {Brain Stimulation},
  volume = {11},
  number = {2},
  pages = {374--389},
  issn = {1935861X},
  doi = {10.1016/j.brs.2017.11.016},
  url = {https://linkinghub.elsevier.com/retrieve/pii/S1935861X17309725},
  urldate = {2024-07-18},
  langid = {english}
}

@article{zrennerShakyGroundTruth2020,
  title = {The Shaky Ground Truth of Real-Time Phase Estimation},
  author = {Zrenner, Christoph and Galevska, Dragana and Nieminen, Jaakko O. and Baur, David and Stefanou, Maria-Ioanna and Ziemann, Ulf},
  date = {2020-07},
  journaltitle = {NeuroImage},
  shortjournal = {NeuroImage},
  volume = {214},
  pages = {116761},
  issn = {10538119},
  doi = {10.1016/j.neuroimage.2020.116761},
  url = {https://linkinghub.elsevier.com/retrieve/pii/S1053811920302482},
  urldate = {2025-01-28},
  langid = {english}
}
	
	\section*{Author contribution}
	
	E.O. and D.K. wrote the manuscript. E.O. performed the mathematical calculations; E.O. designed and performed all data analyses; D.K. supervised the study.
	
	\section*{Competing interests}
	
	The authors declare no competing interests.

\clearpage

    \printtitle{Supplementary Material: Optimal Calibration of the endpoint-corrected Hilbert Transform}{Eike Osmers and Dorothea Kolossa}
	
	\tableofcontents
	
	\clearpage
	
	\appendix
	
	\section{Notation}
	
	\begin{table}[!h]
		\centering
		\caption{Summary of notation.}
		\label{tab:notation}
		\begin{tabular}{ll}
			\toprule
			$C$ & Calibration factor\\
			$D_N$ & Dirichlet kernel of length $N$\\
			$F$ & Effective endpoint complex gain\\
			$F_s$ & Sampling rate (Hz)\\
			$f_0$ & Centre frequency of the target oscillation (Hz)\\
			$f_k$ & Frequency of bin $k$, $f_k = kF_s/L$\\
			$G_+$ & Effective gain on positive-frequency component (main-lobe gain)\\
			$G_-$ & Effective gain on negative-frequency component (leakage gain)\\
			$H(k)$ & Discrete frequency response of the bandpass filter\\
			$H_{\mathrm{eff}}(k)$ & Effective ecHT response, $H_{\mathrm{eff}}(k) = m(k)H(k)$\\
			$k$ & DFT bin index, $k = 0,\dots,L-1$\\
			$L$ & DFT/FFT length ($L \geq N$)\\
			$m(k)$ & Analytic signal mask in the frequency domain\\
			$n$ & Time window index, $n = 0,\dots,N-1$\\
			$N$ & Window length (number of samples)\\
			$x(n)$ & Discrete-time real-valued input signal\\
			$X(k)$ & DFT of the windowed signal\\
			$X^+(k)$ & Analytic DFT signal\\
			$z(n)$ & Analytic signal\\
			$\hat z(n)$ & Analytic signal estimate\\
			$\hat z_\text{end}$ & Estimate of analytic endpoint, $\hat z(N-1)$\\
			$\mathcal{H}$ & Hilbert transform\\
			$\mathcal{P}$ & Set of strictly positive-frequency DFT bins\\
			$\theta(n)$ & True instantaneous phase, $\arg z(n)$\\
			$\hat\theta(n)$ & Instantaneous phase estimate, $\arg \hat z(n)$\\
			$\varphi_0$ & Initial phase of $x$\\
			$\omega_0$ & Centre angular frequency, $\omega_0 = 2\pi f_0/F_s$\\
			$\omega_k$ & Angular frequency of bin $k$, $\omega_k = 2\pi k/L$\\
			\bottomrule
		\end{tabular}
	\end{table}

	\section{ecHT Algorithm} \label{sec:ecHT-alg}
	
	The ecHT is based on the analytic signal representation of a narrow-band real-valued signal. Let $x$ denote a discrete-time signal that is approximately band-limited around a centre frequency $f_0$ (e.g.~a tremor or alpha rhythm). Its analytic signal is defined as
	\begin{align}
		z(n) &= x(n) + \j \mathcal{H} x(n)
	\end{align}
	where $\mathcal{H}$ denotes the Hilbert transform. The Hilbert transform constructs a quadrature counterpart to $x$, where $z$ encodes its instantaneous amplitude and phase. The instantaneous phase is then
	\begin{align}
		%\theta &= \arg z(n) =  \arctantwo \left(\frac{\Im z(n)}{\Re z(n)}\right) = \arctantwo \left(\frac{\Im z(n)}{x(n)}\right)
		\theta &= \arg z(n) =  \arctantwo \left(\Im z(n),\,\Re z(n)\right) = \arctantwo \left(\Im z(n),\, x(n)\right).
	\end{align}
	In the ecHT, this analytic signal is computed in short windows of length $N$ using the DFT. For each new sample, a window $\{x(0), \dots , x(N-1)\}$ is analysed, and only the phase of the last sample $\theta(N-1)$ is used. This endpoint-focused, sliding-window evaluation is what makes the ecHT an online algorithm: it can be implemented recursively and updated sample by sample, without requiring access to future data, while keeping latency low.\\		
	
	The key modification relative to the standard DFT-based Hilbert transform is the introduction of a bandpass filter around the frequency of interest $f_0$ in the frequency domain. After forming the analytic spectrum $X^+$ by zeroing the negative DFT bins and doubling the positive ones, the ecHT multiplies $X^+$ by the frequency response $H$ of a causal narrow-band filter centred at $f_0$:
	\begin{align}
		\hat{Z}(k) = X^+(k) H(k)
	\end{align}
	\begin{algorithm}[!b]
		\caption{endpoint-corrected Hilbert Transform (ecHT)}
		\label{alg:echt}
		\begin{algorithmic}[1]
			\REQUIRE Discrete-time signal $x(n)$, centre frequency $f_0$, bandpass filter $H$
			\ENSURE Estimated instantaneous phase $\hat{\theta}(n)$
			
			\STATE Compute the discrete Fourier transform (DFT) of the input signal
			\begin{align}
				X(k) \leftarrow \mathrm{DFT}\{x(n)\}
			\end{align}
			\vspace{-1.5em}
			\STATE Construct the analytic spectrum $X^+(k)$, with  $\mathcal{P}$ denoting the set of positive DFT frequency indices
			\begin{align}
				X^+(k) \leftarrow 2X(k)\mathbf{1}_{k\in\mathcal{P}}, \quad X^+(0) \leftarrow X(0)
			\end{align}
			\vspace{-1.5em}
			\STATE Apply bandpass filter around the centre frequency $f_0$
			\begin{align}
				\hat{Z}(k) \leftarrow X^+(k) H(k)
			\end{align}
			\vspace{-1.5em}
			\STATE Compute the inverse DFT to obtain the analytic signal estimate
			\begin{align}
				\hat{z}(n) \leftarrow \mathrm{iDFT}\{\hat Z(k)\}
			\end{align}
			\vspace{-1.5em}
			\STATE Extract instantaneous phase
			\begin{align}
				\hat{\theta}(n) \leftarrow \arg \hat{z}(n)
			\end{align}
			\vspace{-1.5em}
			\RETURN $\hat{\theta}(n)$
		\end{algorithmic}
	\end{algorithm}
    In the time domain, this corresponds to convolving the analytic signal with a causal bandpass filter that selectively preserves components near $f_0$ while attenuating others. Importantly, this operation smooths the abrupt jump between the last and first samples introduced by the rectangular window, thereby reducing Gibbs ringing and endpoint distortion. The inverse DFT of $\hat{Z}$ yields the ecHT approximation $\hat{z}$, and the ecHT phase estimate is
	\begin{align}
		\hat{\theta}(n) = \arg \hat{z}(n).
	\end{align}
	Algorithm~\ref{alg:echt} summarises these steps. Using FFT-based implementations, the computational complexity is $\mathcal{O}(L \log L)$ per window, and the filter can be implemented recursively, meaning that only the most recent sample needs updating. For the remainder of the paper, we will denote the true mathematical analytic signal by $z$, and its ecHT-based estimate at the endpoint by $\hat{z}$. We will focus on analysing and calibrating the phase error $\hat{\theta}(N-1) - \theta(N-1)$, which is most relevant for real-time closed-loop applications.
	
	\section{Algorithm Analysis} \label{sec:DFT}
	
	The DFT implicitly assumes a rectangular window around the signal with periodic extension. Hence, a jump discontinuity manifests at the signal's edge in the absence of periodicity, inducing the Gibbs phenomenon and leakage in the time domain. As demonstrated in Fig.~\hyperref[fig:intro]{\ref*{fig:intro}b}, this DFT window forms a Dirichlet kernel in the frequency domain, characterised by the presence of pronounced side lobes. This manifests in spectral leakage. \\
	
	The signal is modelled as a finite-length cosine signal
	\begin{align}
		\begin{split}
			x(n) &= \cos (\omega_0n + \varphi_0)\qquad n=0,\dots,N-1\\
			&= \frac 12 \left(\e^{\j (\omega_0 n + \varphi_0)} + \e^{-\j (\omega_0 n + \varphi_0)}\right)
		\end{split}
	\end{align}
	sampled at $F_s$
	\begin{align}
		\omega_0 = \frac{2\pi f_0}{F_s}
	\end{align}
	
	\subsection{DFT}
	
	The $L$-point DFT of $x$ can be expressed through its values at discrete frequencies $\omega_k = 2\pi k /L$, where $k = 0, \dots, L-1$ corresponds to the DFT bin index. If $L>N$, we assume $x$ to be zero-padded to length $L$.\\
	
	Defining the length-$N$ Dirichlet kernel $D_N$~\cite{oppenheimDiscretetimeSignalProcessing2014} as
	\begin{align}
		D_N(\omega) = \sum_{n=0}^{N-1} \e^{\j\omega n} = \e^{\j \omega \frac{N-1}{2}}\frac{\sin{\frac{N\omega}{2}}}{\sin{\frac{\omega}{2}}},
	\end{align}
	the DFT of the signal can be expressed through the Dirichlet kernel
	\begin{align}
		\label{eq:DFT-transform-DN}
		\begin{split}
			X(k) &= \sum_{n=0}^{N-1} x(n) \e^{-\j \frac{2\pi kn}{L}}, \qquad \omega_k = \frac{2\pi k}{L}\\
			&= \frac 12 \sum_{n=0}^{N-1} \left( \e^{\j (\omega_0 n + \varphi_0)} + \e^{-\j (\omega_0n + \varphi_0)} \right) \e^{-\j \omega_k n}\\
			&= \frac 12 \left[ \e^{\j \varphi_0}\sum_{n=0}^{N-1} \e^{-\j (\omega_k - \omega_0)n} + \e^{-\j \varphi_0} \sum_{n=0}^{N-1} \e^{-\j (\omega_k + \omega_0)n} \right]\\
			&= \frac 12\e^{\j \varphi_0} D_N(\omega_0 - \omega_k) + \frac 12 \e^{-\j \varphi_0} D_N(-\omega_0 - \omega_k)\\
			&= X_+ (k) + X_- (k),
		\end{split}
	\end{align}
	with $X_+$ representing the positive frequency lobe near $\omega_0$ and $X_-$ the negative frequency lobe near $-\omega_0$. The positive main peak of the Dirichlet kernel characterising the analytic harmonic function will appear at $k_0$, which is the index of the bin closest to $\omega_0$.
	\begin{align}
		k_0 = \argmin_k \abs{\omega_0 - \omega_k}
	\end{align}
	
	\subsection{Analytic Signal}
	
	Ideally, the analytic signal would be obtained by completely suppressing the negative-frequency component $X_-$, but $X_-$, which is nonzero almost everywhere, cannot be removed exactly. Instead, the analytic signal is approximated by applying an analytic mask $m$, which sets the negative-frequency DFT bins to zero and doubles the positive-frequency bins. However, residual leakage of $X_-$ into the positive-frequency region remains, introducing a bias in the phase estimate.
	\begin{align}
		m(k) = \begin{cases}
			1 & k = 0\\
			2 & 1 \leq k \leq \lfloor L/2 \rfloor - 1\\
			1 & k = L/2 \ \text{(if L is even)}\\
			0 & \text{otherwise}
		\end{cases}
	\end{align}
	
	\subsection{Filtering}
	
	Let $H$ denote the frequency response of a causal bandpass filter centred on the frequency of interest. Multiplying the analytic spectrum by this response
	\begin{align}
		\hat{Z}(k) = H_\text{eff}(k) X(k), \quad H_\text{eff}(k) = m(k) H(k)
	\end{align}
	applies the desired narrow-band filter in the frequency domain. In practice, $H$ is obtained numerically using standard filter design tools evaluated at DFT frequencies; its closed-form expression is not required for the subsequent analysis.
	
	\subsection{Inverse DFT}
	
	The inverse DFT yields the analytic signal $\hat{z}$.
	\begin{align}
		\hat{z}(n) &= \frac{1}{L} \sum_{k=0}^{L-1} \hat{Z}
		(k) \e^{\j n \omega_k}
		\label{eq:iDFT}
	\end{align}
	The ecHT algorithm is called recursively to estimate the phase of newest sample by evaluating \eqref{eq:iDFT} at the last sample $\hat{z}(N-1) = \hat{z}_\text{end}$.
	\begin{align}
		\begin{split}
			\hat{z}_\text{end} &= \frac{1}{L} \sum_{k=0}^{L-1} H_\text{eff}(k) X_+(k) \e^{\j \omega_k (N-1)} + \frac{1}{L} \sum_{k=0}^{L-1} H_\text{eff}(k) X_-(k) \e^{\j \omega_k (N-1)}\\
			&= \frac{1}{2L} \e^{\j \varphi_0} \sum_{k=0}^{L-1} H_\text{eff}(k) D_N(\omega_0 - \omega_k) \e^{\j \omega_k (N-1)} \\
			&+ \frac{1}{2L} \e^{-\j \varphi_0}  \sum_{k=0}^{L-1} H_\text{eff}(k) D_N(-\omega_0 - \omega_k) \e^{\j \omega_k (N-1)}
		\end{split}
	\end{align}
	We can now define the lumped positive-frequency contribution $\hat Z_+$ and negative-frequency contribution $\hat Z_-$
	\begin{align}
		\hat{Z}_+ &= \frac{1}{2L} \sum_{k=0}^{L-1} H_\text{eff}(k) D_N(\omega_0 - \omega_k) \e^{\j \omega_k (N-1)}\\
		\hat{Z}_- &= \frac{1}{2L} \sum_{k=0}^{L-1} H_\text{eff}(k) D_N(-\omega_0 - \omega_k) \e^{\j \omega_k (N-1)}\\
		\hat{z}_\text{end} &= \e^{\j \varphi_0} \hat{Z}_+ + \e^{-\j \varphi_0} \hat{Z}_-,
	\end{align}
	which incorporate finite window truncation due to the Dirichlet kernel effect $D_N$ and the analytic signal half-spectrum obtained by application of the mask $m$. Additionally, they encapsulate the bandpass filter shape and time index shift to $n=N-1$.
	
	%\subsection{Phase}
	
	%The argument of the analytic function $\hat{z}$ is the phase estimate.
	%\begin{align}
	%	\begin{split}
		%		\hat{\theta}(n) &= \arg \hat{z}(n)\\
		%		&= \omega_0 (N-1) + \varphi_0 + \arg \left\lbrace\sum_{k=0}^{L-1} \abs{H(k)} D_N(\omega_0 - \omega_k) \e^{\j \left[ \omega_k (n-\frac{N-1}{2}) + \varphi_0 - \arg H(k) \right]}\right\rbrace
		%		%&= \frac 12 \omega_0 (N-1) + \operatorname{arctan2} \frac{\sum_k \abs{H(k)} D_N(\omega_0 - \omega_k) \sin (\omega_k (n - \frac{N-1}{2}) - \arg H)}{\sum_k \abs{H(k)} D_N(\omega_0 - \omega_k) \cos (\omega_k (n - \frac{N-1}{2}) - \arg H)}
		%	\end{split}
	%\end{align}
	%A filter with minimal passband width and zero phase around $\omega_0$ is sought. In other words, to make $\hat{\theta}(n)$ closely track the true phase $\omega_0 n + \varphi_0$, the filter should ideally pass only the frequency $\omega_0$ (no gain errors) and introduce no phase shift at that frequency
	%\begin{align}
	%	H(k) = \begin{cases}
		%		1 & \omega_k = \omega_0\\
		%		0 & \text{otherwise}
		%	\end{cases}
	%\end{align}	
	%In practice, an ideal zero-phase response is unattainable for causal filters. When designing a causal filter a trade-off between amplitude and phase accuracy needs to be made. In the following, the contributions to phase error from a non-ideal amplitude response (leakage) and a non-ideal phase response (filter delay) are analysed.		
	
	\subsection{Error Analysis}\label{sec:error-analysis}
	
	When designing a causal ecHT, it is unavoidable that there will be a trade-off between amplitude and phase accuracy. In this section, we characterise the endpoint error induced by finite windowing (i.e.~the Dirichlet kernel), analytic masking, and bandpass filtering. All deterministic endpoint effects will be captured by a single complex gain factor.\\
	
	For a length-$N$ cosine $x(n)=\cos(\omega_0 n+\varphi_0)$ we use the ideal analytic reference at the newest sample
	\begin{align}
		z_\text{end} = z(N-1) = \e^{\j(\omega_0(N-1)+\varphi_0)},
	\end{align}
	and compare the ecHT endpoint estimate $\hat z_\text{end}$ against it.  We define the endpoint gain
	\begin{align}
		F = \frac{\hat z_\text{end}}{z_\text{end}} ,
	\end{align}
	whose argument and magnitude directly encode phase and amplitude errors
	\begin{align}
		\varepsilon_\theta &= \arg \{\hat z_\text{end}\} - \arg \{z_\text{end}\} = \arg F,\\
		\varepsilon_A &= \abs{\hat z_\text{end}}-1 = \abs{F}-1 .
	\end{align}
	In the ideal case $F=1$, so that $\varepsilon_\theta=\varepsilon_A=0$.\\
	
	The ecHT endpoint can be written as a superposition of a positive-frequency and a negative-frequency contribution. After factoring the reference phase $\e^{\j\omega_0(N-1)}$, this yields
	\begin{align}
		\begin{split}
			F &= \frac{\hat{Z}_+\e^{\j \varphi_0}  + \hat{Z}_- \e^{-\j \varphi_0} }{\e^{\j (\omega_0 (N-1) + \varphi_0)}}\\
			&= \e^{-\j \omega_0 (N-1)} \left(\hat{Z}_+ + \hat{Z}_-\e^{-\j 2\varphi_0}\right)\\
			&= G_+ + G_- \e^{-\j2\varphi_0},
		\end{split}
	\end{align}
	where
	\begin{align}
		\begin{split}
			\label{eq:G+}
			G_+ &= \hat Z_+ \e^{-\j\omega_0(N-1)}\\
			&= \frac{1}{2L} \e^{-\j\omega_0(N-1)} \sum_{k=0}^{L-1}
			H_{\mathrm{eff}}(k)\,D_N(\omega_0-\omega_k)\,e^{\j\omega_k(N-1)},
		\end{split}\\
		\begin{split}
			\label{eq:G-}
			G_- &= \hat Z_- \e^{-\j\omega_0(N-1)}\\
			&= \frac{1}{2L} \e^{-\j\omega_0(N-1)} \sum_{k=0}^{L-1}
			H_{\mathrm{eff}}(k)\,D_N(-\omega_0-\omega_k) \e^{\j\omega_k(N-1)}.
		\end{split}
	\end{align}
	$G_+$ is the deterministic complex gain of the implemented ecHT endpoint operator at $\omega_0$, incorporating the effects of a finite window, the analytic mask, and the sampled bandpass response. We denote its unwrapped phase by $\alpha (\omega) = \arg G_+$.~$G_-$ quantifies the residual negative-frequency leakage that is not reduced by the effective analytic mask and filter. Notably, $G_+$ is independent of $\varphi_0$, whereas the leakage term inherits a $\varphi_0$-dependent rotation.

	\subsubsection{Phase}
	Instantaneous phase is taken as the argument of the (estimated) analytic signal
	\begin{align}
		\hat\theta(n)=\arg \hat z(n)=\arctantwo \left(\Im \hat z_\text{end},\, \Re \hat z_\text{end}\right).
		\label{eq:phase-arg}
	\end{align}
	At the endpoint, we have the factorisation $\hat z_\text{end} = z_\text{end}F$, where $F$ collects the deterministic effects of finite windowing, analytic masking, and bandpass filtering and $z_\text{end} = \e^{\j(\omega_0(N-1)+\varphi_0)}$ is the ideal analytic endpoint of the tone. Therefore the endpoint phase estimate satisfies
	\begin{align}
		\hat\theta(N-1)=\arg \hat z_\text{end}
		= \arg\{z_\text{end}F\}
		= \omega_0(N-1)+\varphi_0+\arg F.
	\end{align}
	All systematic distortions enter additively via $\arg F$, enabling the bias-ripple decomposition in Sec.~\ref{sec:bias-ripple-dec}.\\
	
	The literature offers an apparently equivalent alternative that replaces the real part of the sample-based estimate $\Re \hat z(n)$ with the original real sample $x(n)$ in the denominator. While this matches the phase of the ideal analytic signal $\arg z(n)$ (because $\Re z(n)=x(n)$), it is not equivalent for ecHT: the effective gain $F=F_R+\j F_I$ rotates and rescales the quadrature pair, yielding 
	\begin{align}
		\tan \hat\theta_2(N-1)=F_R\tan(\omega_0(N-1)+\varphi_0)+F_I, 
	\end{align}
	i.e.~a $\varphi_0$- and time-dependent phase warping even for a single tone. See Appendix~\ref{sec:phase1 vs phase2} for the full discussion. For this reason, we use \eqref{eq:phase-arg} exclusively throughout.

	\subsubsection{Bias-ripple decomposition} \label{sec:bias-ripple-dec}
	We define the leakage ratio $r$ and the leakage phase offset $\Delta$
	\begin{align}
		r = \frac{\abs{G_-}}{\abs{G_+}},\qquad G_+ = \abs{G_+} \e^{\j\alpha},\qquad G_- = \abs{G_-} \e^{\j\beta},\qquad
		\Delta = \beta-\alpha.
		\label{eq:leakage-ratio}
	\end{align}
	This yields
	\begin{align}
		F(\varphi_0)=\abs{G_+} \e^{\j\alpha}\Bigl(1+r \e^{\j(\Delta-2\varphi_0)}\Bigr).
	\end{align}
	This motivates splitting the total phase error into a deterministic (calibratable) bias $\alpha$ and a $\varphi_0$-dependent leakage ripple $\delta$.
	\begin{align}
		\varepsilon_\theta(\varphi_0)=\arg F=\alpha+\delta_\theta(\varphi_0),
		\qquad
		\delta_\theta(\varphi_0)=\arg F(\varphi_0)-\arg G_+
	\end{align}
	Similarly, the amplitude error decomposes into a deterministic amplitude bias
	controlled by $\abs{G_+}$ and a $\varphi_0$-dependent modulation induced by $G_-$.\\
	
	If the centre frequency shifts from $\omega_0$, $\omega = \omega_0 + \Delta \omega$, then the deterministic bias term changes as
	\begin{align}
		\begin{split}
			\alpha(\omega_0 + \Delta \omega) &\approx \alpha(\omega_0) + \Delta \omega \left. \dv{\alpha}{\omega}\right|_{\omega = \omega_0}\\
			&= \alpha(\omega_0) - \Delta \omega \tau_g (\omega_0).
		\end{split}
		\label{eq:freq-shift}
	\end{align}
	The endpoint group delay $\tau_g$ is defined as
	\begin{align}
		\tau_g (\omega) = - \dv{\omega} \alpha (\omega), \qquad \alpha (\omega) = \arg G_+
		\label{eq:tau_g}
	\end{align}
	as the phase slope of the endpoint gain $G_+$ in \eqref{eq:G+}, not the group delay of the filter $H$ alone.
	
	\begin{theorem}[Deterministic phase-ripple and amplitude bounds, Proof Appendix~\ref{sec:det-bounds}]\label{thm:det-bounds}
		Assume $\abs{G_-}\leq \abs{G_+}$ and define $r\in[0,1]$ and $\alpha=\arg G_+$ as above.
		Let
		\begin{align}
			F(\varphi_0)=G_+ + G_- \e^{-\j2\varphi_0},
			\qquad
			\delta_\theta(\varphi_0)=\arg\{F(\varphi_0)\}-\alpha.
		\end{align}
		Then for all $\varphi_0\in\mathbb{R}$,
		\begin{align}
			\abs{\delta_\theta(\varphi_0)} \leq \arcsin r.
		\end{align}
		Moreover, the amplitude error satisfies the uniform bound
		\begin{align}
			\abs{\varepsilon_A(\varphi_0)} = \big|\abs{F(\varphi_0)}-1\big|
			\leq \big|\abs{G_+}-1\big| + \abs{G_-}.
		\end{align}
	\end{theorem}
	
	A scalar calibration can remove the deterministic bias encoded in $G_+$, whereas $G_-$ produces phase and amplitude ripples that vary with $\varphi_0$ and therefore set a deterministic worst-case floor. Consequently, if the filter successfully suppresses most of the negative-frequency component (e.g.~$r=0.1$), the maximum possible phase error after calibrating $\alpha$ at the newest sample is $\arcsin(0.1) \approx \ang{5.7}$. In practice, typical leakage ratios are on the order of a few percent, yielding only a few degrees of worst-case bias.\\
	
	For $r\ll 1$,
	\begin{align}
		\delta_\theta(\varphi_0)
		= \arg\left\lbrace1+r\e^{\j(\Delta-2\varphi_0)}\right\rbrace
		\approx r\sin(\Delta-2\varphi_0),
		\label{eq:wobble-phase}
	\end{align}
	so the phase error is well-approximated by a constant bias $\alpha$ plus a
	small sinusoidal wobble with period $\pi$ in $\varphi_0$. The endpoint magnitude is
	\begin{align}
		\begin{split}
			|F(\varphi_0)|
			&= \abs{G_+}\sqrt{1+r^2+2r\cos(\Delta-2\varphi_0)}\\
			&\approx \abs{G_+}\Bigl(1+r\cos(\Delta-2\varphi_0)\Bigr),
		\end{split}
		\label{eq:wobble-amplitude}
	\end{align}
	so $\varepsilon_A$ contains a constant part $\abs{G_+}-1$ and a small oscillatory part due to leakage.

	\subsubsection{Error Sources}
	
	The endpoint error can be understood in terms of the complex gain 
	\begin{align}
		F=G_+ + G_- \e^{-\j 2 \varphi_0}. 
	\end{align}
	$G_+$ is independent of the initial phase $\varphi_0$ and acts as a constant complex gain at the frequency of interest. Its argument produces a constant phase bias $\alpha$, whose local frequency sensitivity is governed by the endpoint group delay $\tau_g (\omega_0)$ of \eqref{eq:tau_g}. This component is fully deterministic for a given ecHT design and can therefore be removed by a global complex calibration factor.\\
	
	The second term, $G_- \e^{-\j 2 \varphi_0}$ arises from negative-frequency leakage through the Dirichlet kernel and generates a phase and amplitude ripple around the aforementioned bias. As it depends on $\varphi_0$, it varies from cycle to cycle (even for identical noise conditions) and contributes to the variance of the endpoint error rather than its mean. In the small-leakage regime, this produces a sinusoidal wobble of the phase error in \eqref{eq:wobble-phase} with amplitude proportional to the leakage ratio $r=\abs{G_-}/\abs{G_+}$.\\
	
	This bias-variance dichotomy is made explicit in the MSE decomposition and calibration results of Sections~\ref{sec:calibration} and \ref{sec:derivation_bias_variance}: scalar calibration can eliminate the bias associated with $G_+$, but the variance floor set by $G_-$ and by noise remains, and must be controlled through design choices such as window length, bandwidth and filter order.\\
	
	The analysis depends on the assumed centre frequency $f_0$ because both $H$ and, consequently $G_\pm$ are dependent on the centring of the bandpass. If the true instantaneous frequency is $f(t)=f_0+\Delta f(t)$, then $G_+(f)$ is sampled off-centre and its argument becomes time-varying, resulting in a time-varying phase bias even after applying a fixed calibration. As seen in \eqref{eq:freq-shift}, this mismatch induces a bias proportional to the group delay $\Delta \omega\, \tau_g (\omega_0)$, motivating online tracking of $f_0$ and periodic re-computation of the frequency-specific calibration when drift is present.
	
		\section{Calibration}\label{sec:calibration}
	
	The previous section characterised the deterministic complex gain induced by windowing and filtering at the endpoint $F$. This section shows that a single complex scalar can be used to calibrate the ecHT endpoints to minimise mean-square error with respect to an ideal analytic reference, first in general and then for a single tone.
	
	\subsection{Calibration for General Complex Signals}\label{sec:calibration-general-signal}
	
	Let $Z$ denote the ideal analytic signal endpoint on a window, and let $\hat Z$ denote the corresponding ecHT endpoint estimate on the same window. We model $(Z,\hat Z)$ as a pair of complex-valued random variables capturing variability across repeated windows/trials (e.g.~due to noise, phase, or mild non-stationarity). Expectations below are understood as limits over such repetitions.\\
	
	We introduce a complex scalar calibration factor $C \in \mathbb{C}$ to compensate for both amplitude bias (through $\abs{C}$) and systematic phase offset (through $\arg C$).
	\begin{align}
		\hat z_C(n) = C \hat z(n), \qquad \hat Z_C = C \hat Z
	\end{align}
	Calibration quality is measured by the mean-square error (MSE).
	\begin{align}
		J(C) = \mathbb{E}\left[\,\abs*{C\hat Z - Z}^2\,\right]
		\label{eq:J-def}
	\end{align}
	
	\begin{theorem}[Mean-square optimal scalar calibration, Proof Appendix~\ref{sec:scalar-calibration-general-proof}]
		\label{thm:scalar-calibration}
		Let $Z$, $\hat Z$ be complex random variables with finite second moments and $\mathbb{E}[\abs*{\smash{\hat Z}}^2]>0$. Then $J(C)$ in~\eqref{eq:J-def} is strictly convex in $C\in\mathbb{C}$ and attains its unique minimum at
		\begin{align}
			C_\text{opt} = \frac{\mathbb{E}[\hat Z^{\ast} Z]}{\mathbb{E}[\abs*{\smash{\hat Z}}^2]}.
			\label{eq:Copt-general}
		\end{align}
		The corresponding minimal mean-square error is
		\begin{align}
			J(C_\text{opt}) = \mathbb{E}[\abs{Z}^2]
			- \frac{\abs*{\mathbb{E}[\hat Z^{\ast} Z]}^2}{\mathbb{E}[\abs*{\hat Z}^2]}
			\;\le\; J(1) = \mathbb{E}[\abs*{\hat Z - Z}^2],
			\label{eq:Jmin-general}
		\end{align}
		so an optimally chosen scalar calibration never increases the MSE relative to the uncalibrated estimate.
	\end{theorem}
	
	$C_\text{opt}$ in~\eqref{eq:Copt-general} aligns the phase of $\hat{Z}$ to $Z$ and scales the amplitude, essentially correcting the systematic bias. Thus, $C_\text{opt}$ is the unique scalar calibration that minimises $J(C)$ and has the structure of a complex Wiener filter. It is the scalar linear estimator mapping $\hat Z$ to the best linear mean-square estimate of $Z$.\\
	
	Define the complex correlation coefficient
	\begin{align}
		\rho_{Z\hat Z}
		= \frac{\mathbb{E}[\hat Z^{\ast} Z]}
		{\sqrt{\mathbb{E}[\abs*{\smash{\hat Z}}^2]\mathbb{E}[\abs{Z}^2]}}.
	\end{align}
	
	\begin{corollary}[Correlation form of the minimal MSE, Proof Appendix~\ref{sec:correlation-general-proof}]
		\label{cor:correlation-form}
		With $\rho_{Z\hat Z}$ as above,
		\begin{align}
			J(C_\text{opt}) = \mathbb{E}[\abs{Z}^2]\left(1 - \abs{\rho_{Z\hat Z}}^2\right),
			\label{eq:Jmin-corr}
		\end{align}
		with $0 \leq J(C_\text{opt}) \leq \mathbb{E}[\abs{Z}^2]$. Moreover,
		\begin{itemize}
			\item $J(C_\text{opt}) = 0$ if and only if $\hat Z = C Z$ almost surely (perfect linear predictability);
			\item $J(C_\text{opt}) = \mathbb{E}[\abs{Z}^2]$ if and only if $\hat Z$ and $Z$ are uncorrelated, i.e.\ $\mathbb{E}[\hat Z^{\ast} Z] = 0$.
		\end{itemize}
	\end{corollary}
	
	Geometrically, $C_\text{opt}\hat Z$ is the orthogonal projection of $Z$ onto the one-dimensional subspace spanned by $\hat Z$ in the Hilbert space of complex random variables. The factor $\abs*{\rho_{Z\hat Z}}^2$ is the fraction of analytic-signal energy that can be recovered linearly from the ecHT endpoint, while $1-\abs*{\rho_{Z\hat Z}}^2$ quantifies the loss that is not recoverable due to filtering, windowing, and spectral leakage. Calibration can only recover information present in $\hat Z$; it cannot undo distortions that are uncorrelated with $Z$.\\
	
	It is also helpful to separate bias and variability which is derived in Appendix~\ref{sec:derivation_bias_variance}.
	\begin{align}
		J(C) = \abs{\mathbb{E}[C\hat Z - Z]}^2
		+ \mathbb{E}\left[\abs*{(C\hat Z - Z) - \mathbb{E}[C\hat Z - Z]}^2\right]
	\end{align}
	The bias captures systematic distortions, while the variance measures inconsistency between cycles caused by noise and leakage-induced ripple. The optimal calibration $C_\text{opt}$ minimises both contributions jointly; any further improvement requires a reduction in leakage or noise in the ecHT itself, rather than relying on additional rescaling.\\
	
	In practice, $C_\text{opt}$ is unknown and must be estimated from data. Suppose we have $M$ windows or trials yielding pairs $(Z_i,\hat Z_i)$, $i=0,\dots,M-1$, at a fixed endpoint. A natural empirical estimator is
	\begin{align}
		\hat C_M
		= \frac{\sum_{i=0}^{M-1} \hat Z_i^{\ast} Z_i}{\sum_{i=0}^{M-1} \abs*{\hat Z_i}^2}.
		\label{eq:C-hat}
	\end{align}
	
	\begin{corollary}[Large-sample properties of the empirical calibration, Proof Appendix~\ref{sec:large-sample-conv}] \label{thm:asymptotics}
		Assume that the sequence $\{(Z_i,\hat Z_i)\}_{i\ge1}$ is strictly stationary and weakly dependent, with finite moments of order $2+\delta$ for some $\delta>0$. Then the estimator $\hat C_M$ in~\eqref{eq:C-hat} satisfies:
		\begin{enumerate}
			\item \emph{Consistency:} $\hat C_M \to C_\text{opt}$ almost surely as $M\to\infty$;
			\item \emph{Asymptotic normality:} $\sqrt{M}\,(\hat C_M - C_\text{opt}) \to \mathcal{N}_{\mathbb{C}}(\vec{0},\mat{\Sigma})$ for some covariance matrix $\mat{\Sigma}$.
		\end{enumerate}
		Thus, averaging over more windows or repetitions yields a calibration coefficient that converges to the theoretical optimum at the usual $1/\sqrt{M}$ rate.
	\end{corollary}
	
	Corollary~\ref{thm:asymptotics} provides a general, data-driven scheme for learning a global, complex gain by regressing the ecHT output onto a reference analytic signal. The theorem guarantees that this empirical gain asymptotically converges to the MSE-optimal scalar as the calibration window increases.\ However, in the short-window, near-single-tone regime considered here, where $f_0$ is known or can be accurately estimated, and where an accurate analytic expression for the end\-point error is available, the closed-form calibration in Section~\ref{sec:calibration-single-tone} is both more efficient and more reliable. Data-driven calibration can converge slowly and depend critically on the quality of the reference, potentially degrading endpoint performance if fitted to short windows. Corollary~\ref{thm:asymptotics} can be viewed primarily as providing a theoretical and practical extension for more complex pipelines (e.g.~broader-band or hardware-perturbed signals), where analytic endpoint characterisations are unavailable, but ample calibration data is accessible.
	
	\subsection{Calibration to a Single Tone} \label{sec:calibration-single-tone}
	
	In the single-tone model the only source of randomness across repetitions is the uniform initial phase $\varphi_0 \sim \mathcal{U}(-\pi,\pi)$, so the expectations of Theorem~\ref{thm:scalar-calibration} reduce to averages over $\varphi_0$. In that case, the ecHT endpoint can be written as
	\begin{align}
		F = \frac{\hat z_\text{end}}{z_\text{end}}
		= G_+ + G_- \e^{-\j 2\varphi_0},
		\label{eq:F-standard}
	\end{align}
	where $G_+$ is the complex gain for the positive-frequency component, $G_-$ is the (typically small) gain for the negative-frequency component (leakage), and the ideal analytic endpoint is normalised to $\abs{Z}=1$.\\
	
	Using $\mathbb{E}[\e^{\pm \j2\varphi_0}] = 0$, it follows that
	\begin{align}
		\mathbb{E}[\hat Z^{\ast} Z] = G_+^{\ast}, \qquad
		\mathbb{E}[\abs*{\hat Z}^2] = \abs{G_+}^2 + \abs{G_-}^2.
	\end{align}
	
	\begin{theorem}[Optimal scalar calibration for a single tone, Proof Appendix~\ref{sec:scalar-calibration-proof}] 		\label{thm:single-tone-calibration}
		Under the assumptions above, the MSE-optimal scalar calibration factor~\eqref{eq:Copt-general} reduces to
		\begin{align}
			C_\text{opt} = \frac{G_+^{\ast}}{\abs{G_+}^2 + \abs{G_-}^2},
			\label{eq:Copt-single-tone}
		\end{align}
		and the corresponding minimal MSE is
		\begin{align}
			J(C_\text{opt}) = \frac{\abs{G_-}^2}{\abs{G_+}^2 + \abs{G_-}^2}.
			\label{eq:Jmin-single-tone}
		\end{align}
		In other words, the irreducible error after optimal calibration equals the fraction of negative-frequency leakage power in the total ecHT endpoint power.
	\end{theorem}
	
	From~\eqref{eq:Copt-single-tone}, \eqref{eq:Jmin-single-tone} the correlation coefficient is also obtained
	\begin{align}
		\rho_{Z\hat Z} = \frac{G_+^*}{\sqrt{\abs{G_+}^2 + \abs{G_-}^2}},
		\qquad
		\abs{\rho_{Z\hat Z}}^2 = \frac{\abs{G_+}^2}{\abs{G_+}^2 + \abs{G_-}^2},
	\end{align}
	where $G_+^*$ is the complex conjugate of $G_+$. So, in the single-tone case the recoverable signal fraction is determined by the ratio of main-lobe gain to total power, and calibration cannot correct the portion of energy that has leaked from the negative frequency.

	\subsection{Calibration Procedure} \label{subsec:calibration}
	
	\begin{algorithm}[!ht]
		\caption{Abstract calibration procedure of the ecHT (with optional $f_0$ estimation)}
		\label{alg:calib}
		\begin{algorithmic}[1]
			\REQUIRE Window length $N$, sampling rate $F_s$, centre frequency $f_0$, bandpass $(l_\text{freq}, h_\text{freq})$, filter order, FFT length $L$, data window $x[0{:}N-1]$.
			\ENSURE Global calibration gain $C_\text{opt}$
			
			\IF{$f_0$ is unstable}
			\STATE Estimate $f_0$ using a suitable spectral estimator
			\ENDIF
			
			\STATE Design $f_0$-centred bandpass $H$
			%\STATE Combined filter \& Hilbert response $H_\text{eff}(k) \leftarrow m(k)\,H(k)$
			%\STATE Compute spectra $X_+(k)$ and $X_-(k)$ at $f_0$ using length-$N$ Dirichlet kernel.
			%\STATE Compute endpoint contributions $\hat Z_\pm \leftarrow \sum_k H_\text{eff}(k) X_\pm(k) \mathrm{e}^{\mathrm{i}\omega_k (N-1)} / L$
			\STATE Compute $G_\pm$ from \eqref{eq:G+}, \eqref{eq:G-}
			\STATE Compute MSE-optimal calibration gain
			\begin{align}
				C_\text{opt}
				\leftarrow \frac{G_+^*}{\abs{G_+}^2 + \abs{G_-}^2},
			\end{align}
			and residual error
			\begin{align}
				J_\text{opt}
				\leftarrow \frac{\abs{G_-}^2}{\abs{G_+}^2 + \abs{G_-}^2}.
			\end{align}
			\STATE Apply $C_\text{opt}$ as a global complex multiplier to the ecHT output $\hat z$
			\begin{align}
				\hat z_C \leftarrow C_\text{opt} \hat z
			\end{align}
		\end{algorithmic}
	\end{algorithm}
	
	Ideally, calibration is performed on a stable oscillation. Algorithm~\ref{alg:calib} outlines the procedure. If the target frequency changes, $C$ may need to be recomputed. For a fixed choice of window length $N$, sampling rate $F_s$, bandpass parameters $(l_\text{freq}, h_\text{freq})$ and centre frequency $f_0$, the ecHT endpoint error for a finite-length cosine can be written as
	\begin{align}
		F(\varphi_0)
		= G_+ + G_- \e^{-2\j\varphi_0},
	\end{align}
	where $\varphi_0$ is the (unknown) initial phase and $G_\pm \in \mathbb{C}$ are the phase-independent gains of the positive and negative frequency components induced by ecHT. The MSE-optimal global scalar $C_\text{opt}$ is obtained in Theorem~\ref{thm:single-tone-calibration}
	\begin{align}
		C_\text{opt}
		= \frac{\mathbb{E}[F^*]}{\mathbb{E}[\abs{F}^2]}
		= \frac{G_+^*}{\abs{G_+}^2 + \abs{G_-}^2} \approx \frac{1}{G_+},
	\end{align}
	and the corresponding minimal residual error is
	\begin{align}
		J(C_\text{opt})
		= \frac{\abs{G_-}^2}{\abs{G_+}^2 + \abs{G_-}^2}.
	\end{align}
	Thus, the entire calibration is determined by these two complex numbers. $G_+$ describes the desired, phase-aligned response at $f_0$ and $G_-$ quantifies the negative-frequency leakage that cannot be removed by a single scalar. The complexity of the calibration $\mathcal{O}(L)$ is negligible compared to the ecHT's complexity of $\mathcal{O}(L \log L)$.
	
	\section{Metrics}\label{sec:experiments}
	
	To summarise the phase-error distribution between the causal ecHT estimate 
	$\hat \theta$ and an offline narrow-band reference $\theta$, standard phase-synchrony statistics on the wrapped phase difference $\Delta \theta = \operatorname{wrap} (\hat \theta - \theta)$ are used. The phase-locking value (PLV)
	\begin{align}
		\operatorname{PLV} = \abs{\mathbb{E} [\e^{\j \Delta \theta}]},
	\end{align}
	measures the concentration of $\Delta \theta$	on the unit circle ($\operatorname{PLV} \to 1$ for a fixed phase relation; $\operatorname{PLV} \to 0$ for broadly dispersed errors)~\cite{lachauxMeasuringPhaseSynchrony1999}.\\ 
	
	The phase-lag index (PLI) captures the directional asymmetry of the phase-difference distribution~\cite{stamPhaseLagIndex2007}. Here, $\Delta \theta >0$ indicates that the estimate $\hat \theta$ leads the reference and $\Delta \theta < 0$ indicates lag; thus $\operatorname{PLI} \approx 0$ implies a symmetric error distribution, whereas larger PLI imply errors of one sign, independent of the error concentration captured by PLV.
	\begin{align}
		\operatorname{PLI} =  \abs{\mathbb{E} [\operatorname{sign} (\Delta \theta)]}
	\end{align}
	
	% \subsection{Simulation} \label{sec:simulated-experiments}
	
	% Computational cost was evaluated on two processors in Tab.~\ref{tab:ecHT_runtime_cpu}. Compared to a standard FFT-based Hilbert transform with the same length, the ecHT produced only a modest overhead ($\sim\!\! 1.1\dots 1.8\times$, median $\sim\!\! 1.4\times$), reflecting the additional bandpass and calibration operations while remaining well within real-time constraints.
	
	% \begin{table}[!h]
	% 	\centering
	% 	\caption{Computational cost of block-wise ecHT.}
	% 	\label{tab:ecHT_runtime_cpu}
	% 	\begin{tabular}{@{}lSS@{}}
	% 		\toprule
	% 		& \multicolumn{2}{c}{CPU} \\
	% 		\cmidrule(lr){2-3}
	% 		& {i7-14700KF} & {i5-1335U} \\
	% 		\midrule
	% 		\multicolumn{3}{l}{\textit{Runtime} [\si{\micro\second}]} \\
	% 		\addlinespace[0.3em]
	% 		$N = \num{256}$
	% 		& 17.0 & 38.9 \\
	% 		$N = \num{1024}$
	% 		& 24.8 & 49.4 \\
	% 		$N = \num{16384}$
	% 		& 204.1 & 389.1 \\
	% 		\addlinespace[0.6em]
	% 		\multicolumn{3}{l}{\textit{Amortised cost} [\si{\nano\second}/sample]} \\
	% 		\addlinespace[0.3em]
	% 		Min
	% 		& 12.7 & 24.3 \\
	% 		Median
	% 		& 16.7 & 33.3 \\
	% 		Max
	% 		& 66.4 & 140.1 \\
	% 		\addlinespace[0.6em]
	% 		Overhead vs.~HT [$\times$]
	% 		& \multicolumn{2}{r}{1.4 (median)} \\
	% 		\bottomrule
	% 	\end{tabular}
	% \end{table}

	\section{Practical recommendations} \label{sec:practical-rec}
	
	As Liufu et al.~\cite{liufuOptimizingRealtimePhase2025} and Zrenner et al.~\cite{zrennerShakyGroundTruth2020} observed for the ecHT, different parameter settings result in different errors, and a compromise usually needs to be made between a constant phase bias, which is determined by the group delay of the bandpass filter, and cycle-to-cycle ripple/variance, which is caused by finite-window spectral leakage, (including negative-frequency leakage). Global calibration can remove most static bias, but cannot eliminate the variance floor induced by leakage, noise and frequency drift.\\
	
	Increasing the window length $N$ reduces leakage and improves the DFT resolution $\Delta f=F_s/N$, giving alignment of $f_0$ and its nearest bin. In practice, the cleanest endpoint behaviour is achieved when the window spans approximately an integer number of cycles, i.e.
	\begin{align}
		N\frac{f_0}{F_s} - \operatorname{round} \left(N\frac{f_0}{F_s}\right) \approx 0.
	\end{align}	
	However, using longer windows effectively treats the oscillation as stationary over that interval; if the frequency drifts, the error increases. For non-stationary rhythms, Liufu et al. report the best performance around one-cycle windows~\cite{liufuOptimizingRealtimePhase2025}, whereas we noted best performance with two-cycle windows. Also, tapered windows such as the Hann window should be avoided for online endpoint use, as they down-weight the newest samples.
	
	\textbf{Recommendation} Use low-integer-cycle windows when the frequency is variable; otherwise use longer, integer-cycle windows to reduce leakage.\\
	
	A narrow filter bandwidth can best suppress off-target frequency components. However, bandwidth can be seen as a \textit{robustness knob}: narrower is not always better, since narrowing the passband increases group delay, which manifests as a constant phase offset. This offset can be calibrated if $f_0$ is stable, but it also makes the estimate more sensitive to detuning/drift, with phase error growing roughly according to 
	\begin{align}
		\varepsilon_\theta \approx -\Delta \omega \tau_g (\omega_0).
	\end{align}	
	\textbf{Recommendation} Use the narrowest possible bandwidth that still covers the expected $f_0$-variability. If the true frequency is uncertain or drifting, use a wider filter and/or actively track $f_0$, updating the filter and calibration as necessary.\\
	
	Higher-order filters generally increase group delay, thereby steepening the local phase slope around $\omega_0$. Consequently, even a small centre-frequency mismatch produces a phase error (Fig.~\hyperref[fig:simulation]{\ref*{fig:simulation}c}). High filter orders can also make the passband phase less locally linear, reducing predictability.
	
	\textbf{Recommendation} Use second order filters and avoid increasing the order unless you can tolerate extra delay/bias and increased sensitivity to detuning.\\
	
	The performance of the ecHT depends critically on centring the filter at the true oscillation frequency. Any mismatch samples the filter's phase response at the incorrect point, resulting in a phase error proportional to the group delay. This is an unavoidable limitation of causal filtering: even after calibration removes the offset at the nominal $f_0$, frequency drift creates a time-varying bias. For non-stationary oscillations, track $f_0$, and update the filter centre frequency and calibration as needed~\cite{corcoranReliableAutomatedMethod2018}.
	
	\textbf{Recommendation} Continuously re-estimate/track $f_0$ in non-stationary settings; treat calibration as frequency-specific, refreshing it when $f_0$ changes.\\
	
	These choices are secondary, but still important for numerical fidelity. Higher sampling rates generally improve the spectral representation. Zero-padding increases the density of DFT bins, which can help to evaluate $H$ closer to the true $f_0$, but it does not change the fundamental leakage-delay trade-off. Finally, Liufu et al. as well as Fig.~\hyperref[fig:simulation]{\ref*{fig:simulation}e} display little difference across common IIR families (e.g.~Butterworth, Chebyshev, Cauer) for the ecHT performance~\cite{liufuOptimizingRealtimePhase2025}; as their differences become more prominent at higher filter orders. Butterworth is a reasonable default.
	
	\textbf{Recommendation} Use zero-padding only for finer frequency grid evaluation. Prefer a sufficient $F_s$ and default to Butterworth unless you have a specific reason not to.\\
	
	We recommend the following calibration practices.
	\begin{itemize}
		\item \textbf{Stable rhythms} For signals with a time-invariant centre frequency, compute $C_\text{opt}$ once during system initialisation and apply it globally.
		
		\item \textbf{Non-stationary rhythms} For signals with drifting frequency, re-estimate $f_0$ depending on your signal properties and update $C_\text{opt}$ accordingly.
		
		\item \textbf{Multiple target frequencies} For applications requiring simultaneous tracking of multiple rhythms, compute frequency-specific calibration factors $C_{\text{opt}}^{f_{i}}$ for each band and apply them independently to parallel ecHT channels.
	\end{itemize}
	
	%\subsection{DFT Grid}
	
	%Evaluate the filter response exactly at the FFT/DFT bin frequencies used by the transform. Avoid rounding bin frequencies (by using $\lceil f_k \rceil$): this creates a frequency-grid mismatch that manifests as an avoidable phase bias. In typical settings, removing this rounding substantially reduces the ecHT mean phase error, especially if the window length is not an integer multiple of the sampling rate or the sampling rate varies. 
	
	\section{Proofs \& Derivations}
	
	\subsection{Deterministic phase-ripple and amplitude bounds} \label{sec:det-bounds}

	\begin{proof}[Proof of Theorem~\ref{thm:det-bounds}]
		Fix $\varphi_0\in\mathbb{R}$.
		\begin{align}
			F(\varphi_0)=G_+ + G_-\e^{-\j2\varphi_0}
			=G_+\bigl(1+u(\varphi_0)\bigr),
			\qquad
			u(\varphi_0)=\frac{G_-}{G_+}\e^{-\j2\varphi_0}
		\end{align}
		Then $\abs{u(\varphi_0)}=\abs{G_-}/\abs{G_+}=r\in[0,1]$, with $\alpha=\arg G_+$,
		\begin{align}
			\delta_\theta(\varphi_0)=\arg \{F(\varphi_0)\}-\alpha=\arg \{1+u(\varphi_0)\}.
		\end{align}
		
		\paragraph{Phase-ripple bound.}
		The goal is to bound $\abs{\delta_\theta(\varphi_0)}$ uniformly in the unknown	initial phase $\varphi_0$. Since $\delta_\theta(\varphi_0)$ depends on $\varphi_0$	only through $u(\varphi_0)$, it suffices to characterise the set of all possible values of $\arg\{1+u\}$.\\
		
		Because $\abs{u(\varphi_0)}=r$ is fixed while the phase of $u(\varphi_0)$ varies with $\varphi_0$, we may write $u(\varphi_0)=r\e^{\j\psi}$ for some $\psi\in\mathbb{R}$. Hence
		\begin{align}
			\delta_\theta(\varphi_0)=\arg \left\lbrace 1+r\e^{\j\psi}\right\rbrace,
		\end{align}
		and a uniform bound in $\varphi_0$ follows from a bound on
		$\max_{\psi\in\mathbb{R}} \abs{\arg\{1+r\e^{\j\psi}\}}$.\\
		
		For notational convenience, we define
		\begin{align}
			w = 1+r\e^{\j\psi} = (1+r\cos\psi) + \j(r\sin\psi).
		\end{align}
		If $0\leq r<1$, then $\Re w=1+r\cos\psi\geq 1-r>0$, hence $\arg w \in[-\pi/2,\pi/2]$ and $\sin(\cdot)$ is strictly increasing on $[0,\pi/2]$, and thus bounding $\abs{\arg w}$ is equivalent to bounding $\abs{\sin(\arg w)}$:
		\begin{align}
			\abs{\arg w} \leq \arcsin r \quad \Leftrightarrow \quad
			\abs{\sin(\arg w)}\leq r.
			\label{eq:sinarg1}
		\end{align}
		
		Using $\sin(\arg w)=\Im \{w\}/\abs{w}$ we get
		\begin{align}
			\abs{\sin(\arg w)}
			=\frac{\abs{r\sin\psi}}{\sqrt{(1+r\cos\psi)^2+(r\sin\psi)^2}}
			=\frac{r\abs{\sin\psi}}{\sqrt{1+r^2+2r\cos\psi}}.
			\label{eq:sinarg2}
		\end{align}
		Substituting \eqref{eq:sinarg2} into \eqref{eq:sinarg1}, then $\abs{\sin(\arg w)}\leq r$ is equivalent to
		\begin{align}
			\abs{\sin\psi} \leq \sqrt{1+r^2+2r\cos\psi}.
		\end{align}
		If $r=0$ the bound is trivial. Squaring both sides gives
		\begin{align}
			\begin{split}
				\sin^2\psi &\leq 1+r^2+2r\cos\psi\\
				1-\cos^2\psi &\leq 1+r^2+2r\cos\psi\\
				(\cos\psi+r)^2 &\geq 0,
			\end{split}
		\end{align}
		which is always true. Hence $\abs{\sin(\arg w)}\leq r$, and since
		$\arg w\in[-\pi/2,\pi/2]$ we conclude
		\begin{align}
			\abs{\arg w} \leq \arcsin r.
		\end{align}
		Recalling $\delta_\theta(\varphi_0)=\arg w$ yields
		\begin{align}
			\abs{\delta_\theta(\varphi_0)} \leq \arcsin r.
		\end{align}
		
		For the boundary case $r=1$, the same inequality implies $\abs{\arg w}\leq \arcsin(1)=\pi/2$ whenever $w\neq 0$; at the isolated phases where $w=0$ (equivalently $F(\varphi_0)=0$) the argument is undefined and the bound is understood in the limiting sense.
		
		\paragraph{Amplitude bound.}
		From $F(\varphi_0)=G_+ + G_-\e^{-\j2\varphi_0}$ and the reverse triangle inequality,
		\begin{align}
			\bigl|\abs{F(\varphi_0)}-\abs{G_+}\bigr| \leq \abs{G_-}.
		\end{align}
		Then applying the triangle inequality on the real line gives
		\begin{align}
			\bigl|\abs{F(\varphi_0)}-1\bigr|
			\leq \bigl|\abs{F(\varphi_0)}-\abs{G_+}\bigr| + \bigl|\abs{G_+}-1\bigr|
			\leq \abs{G_-} + \bigl|\abs{G_+}-1\bigr|.
		\end{align}
		Since $\varepsilon_A(\varphi_0)=\abs{F(\varphi_0)}-1$, this is exactly the claimed bound.
	\end{proof}

	\subsection{Mean-square optimal scalar calibration}
	\label{sec:scalar-calibration-general-proof}
	
	\begin{proof}[Proof of Theorem~\ref{thm:scalar-calibration}]
		Here and throughout, expectations are taken with respect to the distribution of window-level endpoint estimates (across windows, trials, or repetitions), not with respect to time within a single window. Assume finite second moments of $Z$, $\hat Z$
		\begin{align}
			\mathbb{E}[\abs*{Z}^2] < \infty, \qquad
			\mathbb{E}[\abs*{\smash{\hat Z}}^2] < \infty, \qquad
			\mathbb{E}[\abs*{\smash{\hat Z}}^2] > 0
		\end{align}
		which excludes the degenerate case $\hat Z = 0$ almost surely.\\
		
		The squared magnitude in \eqref{eq:J-def} is expanded as
		\begin{align}
			\begin{split}
				J(C)
				&= \mathbb{E}\left[(C \hat{Z} - Z)(C \hat{Z} - Z)^*\right] \\
				&= \mathbb{E}[\abs*{C \hat{Z}}^2]
				- \mathbb{E}[C \hat{Z} Z^*]
				- \mathbb{E}[C^* \hat{Z}^* Z]
				+ \mathbb{E}[\abs{Z}^2].
			\end{split}
		\end{align}
		Linearity of expectation and the fact that $C$ is deterministic lead to
		\begin{align}
			J(C)
			= \abs{C}^2 \,\mathbb{E}[\abs*{\hat{Z}}^2]
			- C\,\mathbb{E}[\hat{Z} Z^*]
			- C^*\,\mathbb{E}[\hat{Z}^* Z]
			+ \mathbb{E}[\abs{Z}^2].
		\end{align}
		Introducing the shorthand
		\begin{align}
			\kappa = \mathbb{E}[\abs*{\hat{Z}}^2], 
			\quad
			\nu = \mathbb{E}[\hat{Z}^* Z],
			\quad
			\xi = \mathbb{E}[\abs{Z}^2], \quad
		\end{align}
		gives the quadratic form
		\begin{align}
			J(C) = \kappa \abs{C}^2 - \nu^* C - \nu C^* + \xi.
			\label{eq:J-quadratic2}
		\end{align}
		Because $\kappa>0$, the function $J$ is strictly convex in $C$. Completion of the square yields
		\begin{align}
			\begin{split}
				\kappa\abs{C}^2 - \nu^* C - \nu C^*
				+ \frac{\abs{\nu}^2}{\kappa} &= \kappa \left[CC^* - C\frac{\nu^*}{\kappa} - C^* \frac{\nu}{\kappa} + \frac{\nu}{\kappa} \frac{\nu^*}{\kappa}\right] \\
				&= \kappa\abs{C - \frac{\nu}{\kappa}}^2.
			\end{split}
		\end{align}
		Comparing with \eqref{eq:J-quadratic2} using $\kappa\in\mathbb{R}$ gives
		\begin{align}
			J(C)
			= \kappa\abs{C - \frac{\nu}{\kappa}}^2
			+ \xi - \frac{\abs{\nu}^2}{\kappa}.
		\end{align}
		The term $\kappa\abs{C - \nu/\kappa}^2$ is nonnegative and equals zero if and only if
		\begin{align}
			C = C_\text{opt} = \frac{\nu}{\kappa}
			= \frac{\mathbb{E}[\hat{Z}^* Z]}{\mathbb{E}[\abs*{\hat{Z}}^2]}.
		\end{align}
		This expression coincides with \eqref{eq:Copt-general} and is the unique minimiser of $J$. Evaluating $J$ at $C = C_\text{opt}$ gives
		\begin{align}
			J\left(C_\text{opt}\right)
			= \xi - \frac{\abs{\nu}^2}{\kappa}
			= \mathbb{E}[\abs{Z}^2]
			- \frac{\abs{\mathbb{E}[\hat{Z}^* Z]}^2}{\mathbb{E}[\abs*{\hat{Z}}^2]},
		\end{align}
		which is \eqref{eq:Jmin-general}. Since $C_\text{opt}$ is the global minimiser, $J(C_\text{opt}) \leq J(C)\ \forall\ C \in \mathbb{C}$.
	\end{proof}
	
	\subsection{Correlation form of the minimal MSE} \label{sec:correlation-general-proof}
	
	\begin{proof}[Proof of Corollary~\ref{cor:correlation-form}]
		By the definition of $\rho_{Z\hat{Z}}$,
		\begin{align}
			\abs*{\mathbb{E}[\hat{Z}^* Z]}^2
			= \abs{\rho_{Z\hat{Z}}}^2
			\,\mathbb{E}[\abs*{\hat{Z}}^2]\,\mathbb{E}[\abs{Z}^2].
		\end{align}
		Substitution into \eqref{eq:Jmin-general} yields
		\begin{align}
			\begin{split}
				J\left(C_\text{opt}\right)
				&= \mathbb{E}[\abs{Z}^2]
				- \frac{\abs{\rho_{Z\hat{Z}}}^2
					\,\mathbb{E}[\abs*{\hat{Z}}^2]\,\mathbb{E}[\abs{Z}^2]}
				{\mathbb{E}[\abs*{\hat{Z}}^2]} \\
				&= \mathbb{E}[\abs{Z}^2]\left(1 - \abs{\rho_{Z\hat{Z}}}^2\right),
			\end{split}
		\end{align}
		which is \eqref{eq:Jmin-corr}. The Cauchy-Schwarz inequality implies
		$\abs*{\rho_{Z\hat{Z}}} \leq 1$, hence $0 \leq 1 - \abs*{\rho_{Z\hat{Z}}}^2 \leq 1$ and
		\begin{align}
			0 \leq J\left(C_\text{opt}\right) \leq \mathbb{E}[\abs{Z}^2].
		\end{align}
		Equality $J(C_\text{opt}) = 0$ holds if and only if $\abs{\rho_{Z\hat{Z}}} = 1$, which is equivalent to $\hat{Z} = C Z$ almost surely for some constant $C$. Equality $J(C_\text{opt}) = \mathbb{E}[\abs{Z}^2]$ holds if and only if	$\rho_{Z\hat{Z}} = 0$, i.e.\ $\hat{Z}$ and $Z$ are uncorrelated.
	\end{proof}
	
	Corollary~\ref{cor:correlation-form} shows that, after calibration, the minimal achievable mean-square error is entirely determined by the squared magnitude of the complex correlation coefficient between the ecHT endpoint $\hat{Z}$	and the ideal analytic signal $Z$. The remaining error is equal to the signal power scaled by $(1-\abs*{\rho_{Z\hat{Z}}}^2)$, which quantifies the fraction of signal energy that is fundamentally unrecoverable by any linear calibration. Perfect correlation $\abs*{\rho_{Z\hat{Z}}}^2 = 1$ implies that calibration can fully correct the ecHT, whereas zero correlation means that the ecHT output contains no usable information about the true signal. Thus, correlation (not calibration) sets the fundamental accuracy limit of the cecHT, and all sources of distortion impact performance only insofar as they reduce this correlation.\\
	
	Geometrically, the calibrated estimate $C_\text{opt}\hat Z$ is the orthogonal projection of $Z$ onto the one-dimensional subspace spanned by $\hat Z$ in the Hilbert space of complex random variables, and $J(C_\text{opt})$ equals the squared norm of the irreducible projection residual. Equivalently, the factor $\abs*{\rho_{Z\hat Z}}^2$ represents the fraction of the analytic signal energy that is linearly recoverable from the ecHT output, while $1 - \abs*{\rho_{Z\hat Z}}^2$ quantifies the fundamental energy loss induced by filtering, windowing, and spectral leakage.
	
	\subsection{Bias-variance decomposition} \label{sec:derivation_bias_variance}
	
	For any scalar calibration $C$, the mean-square error admits the decomposition
	\begin{align}
		J(C) 
		= \abs*{\mathbb{E}[C\hat Z - Z]}^2 
		+ \mathbb{E}\left[ \abs*{(C\hat Z - Z) - \mathbb{E}[C\hat Z - Z]}^2\right],
	\end{align}
	i.e.~into a squared bias term and a variance term. The bias term captures the deterministic distortions caused by filter delay, amplitude misscaling, and systematic phase offsets.~The variance term, on the other hand, measures the inconsistency in the estimate from one cycle to the next, which is caused by noise, spectral leakage and phase-dependent interference between positive and negative frequencies. The optimal calibration $C_\text{opt}$ minimises both contributions jointly: it aligns the ecHT estimate with the ideal analytic signal as closely as possible on average, while also yielding the smallest possible endpoint variability of all scalar calibrations. Consequently, any further improvement beyond calibration must reduce either leakage or noise in the ecHT itself, as post-hoc scaling cannot reduce the residual variance below this limit.\\
	
	Define the complex-valued endpoint error
	\begin{align}
		e = C\hat Z - Z,
	\end{align}
	so that
	\begin{align}
		J(C) = \mathbb{E} [\abs*{C\hat Z - Z}^2 ] = \mathbb{E} [\abs{e}^2 ].
	\end{align}
	Let
	\begin{align}
		\mu = \mathbb{E}[e] = \mathbb{E}[C\hat Z - Z]
	\end{align}
	denote the mean error. Then $e$ can be written as
	\begin{align}
		e = \mu + (e - \mu),
	\end{align}
	where $e - \mu$ is a zero-mean fluctuation. Substituting this into $\abs{e}^2$ yields
	\begin{align}
		\begin{split}
			\abs{e}^2
			&= \abs{\mu + (e - \mu)}^2 \\
			&= (\mu + (e - \mu))(\mu + (e - \mu))^* \\
			&= (\mu + (e - \mu))(\mu^* + (e - \mu)^*) \\
			&= \abs{\mu}^2 + \abs{e - \mu}^2 + \mu^*(e - \mu) + \mu(e - \mu)^*.
		\end{split}
	\end{align}
	Taking expectations on both sides gives
	\begin{align}
		\mathbb{E}[\abs{e}^2]
		&= \mathbb{E}[\abs{\mu}^2] 
		+ \mathbb{E}[\abs{e - \mu}^2] 
		+ \mathbb{E}[\mu^*(e - \mu)] 
		+ \mathbb{E}[\mu(e - \mu)^*].
	\end{align}
	Since $\mu$ is deterministic, $\mathbb{E}[\abs{\mu}^2] = \abs{\mu}^2$. Moreover,
	\begin{align}
		\mathbb{E}[\mu^*(e - \mu)] = \mu^*\,\mathbb{E}[e - \mu] 
		= \mu^*(\mathbb{E}[e] - \mu) 
		= \mu^*(\mu - \mu) = 0,
	\end{align}
	and similarly $\mathbb{E}[\mu(e - \mu)^*] = 0$, because $\mathbb{E}[e - \mu] = 0$ by construction. Hence
	\begin{align}
		\mathbb{E}[\abs{e}^2] = \abs{\mu}^2 + \mathbb{E}[\abs{e - \mu}^2].
	\end{align}
	Substituting back $e = C\hat Z - Z$ and $\mu = \mathbb{E}[C\hat Z - Z]$ yields the bias-variance decomposition
	\begin{align}
		\begin{split}
			J(C)
			&= \mathbb{E}\left[\abs*{C\hat Z - Z}^2\right]\\
			&= \abs*{\mathbb{E}[C\hat Z - Z]}^2
			+ \mathbb{E}\left[\abs*{(C\hat Z - Z) - \mathbb{E}[C\hat Z - Z]}^2\right].
		\end{split}
	\end{align}

	\subsection{Large-sample properties of the empirical calibration} \label{sec:large-sample-conv}
	
	We estimate the optimal scalar calibration
	\begin{align}
		C_\text{opt} = \frac{\mathbb{E}[\hat Z^* Z]}{\mathbb{E}[\abs*{\hat Z}^2]}
	\end{align}
	from $M$ consecutive windows. Because windows are extracted from a time series, the resulting sequence $\{(Z_i,\hat Z_i)\}_{i\geq 1}$ is typically not iid; neighbouring windows may be correlated. The goal is to state conditions under which the ratio estimator
	\begin{align}
		\hat C_M = \frac{\frac 1M \sum_{i=0}^{M-1} \hat Z_i^* Z_i}{\frac 1M \sum_{i=0}^{M-1} \abs*{\hat Z_i}^2}
	\end{align}
	is (1) consistent and (2) asymptotically normal as $M\to \infty$.\\
	
	Define
	\begin{align}
		A_i = \hat Z_i^* Z_i,
		\quad
		B_i = \abs*{\hat Z_i}^2,
		\quad
		\bar A_M = \frac{1}{M}\sum_{i=0}^{M-1} A_i,
		\quad
		\bar B_M = \frac{1}{M}\sum_{i=0}^{M-1} B_i,
	\end{align}
	so that $\hat C_M = \bar A_M / \bar B_M$ and $C_\text{opt} = \mathbb{E}[A_1] / \mathbb{E}[B_1]$.\\
	
	The assumptions below are intended for window-level features $(Z_i, \hat Z_i)$, not raw samples: each window aggregates many samples, so correlations between widely separated windows are typically weaker than correlations between adjacent raw samples.
	\begin{enumerate}
		\item[(i)] \textbf{Strict stationarity} $\{(A_i,B_i)\}$ is invariant under time shifts (e.g.~$\mathbb{E}[A_i]=\mathbb{E}[A_1]$ and $\mathbb{E}[B_i]=\mathbb{E}[B_1]$). This ensures $\mathbb{E}[A_i]$ and $\mathbb{E}[B_i]$ do not depend on $i$, i.e.~the target $C_\text{opt}$ is time-invariant.
		\item[(ii)] \textbf{Ergodicity} Time averages converge to ensemble averages. Concretely, ergodicity delivers a law of large numbers for dependent data, so $\bar A_M \to \mathbb{E}[A_1]$ and $\bar B_M \to \mathbb{E}[B_1]$ even without independence.
		\item[(iii)] \textbf{Weak dependence/Strong mixing} Dependence between windows decays with lag. This is not needed to define $C_\text{opt}$ but it is what enables a central limit theorem (CLT) and therefore uncertainty quantification for $\hat C_M$.
		\item[(iv)] \textbf{Finite moments and non-degeneracy} $\mathbb{E}[A_1] < \infty$, $\mathbb{E}[B_1] < \infty$, and $\mathbb{E}[B_1] > 0$. The first two conditions ensure integrability. The last condition prevents a denominator in the derivation from collapsing and guarantees the ratio map is well behaved near the limit.
		\item[(v)] \textbf{Finite higher moments} $\mathbb{E}[\abs{A_1}^{2+\delta}] < \infty$ and $\mathbb{E}[\abs{B_1}^{2+\delta}] < \infty$	for some $\delta>0$.
	\end{enumerate}
	Define the empirical calibration based on $M$ windows by
	\begin{align}
		\hat C_M
		= \frac{\frac{1}{M}\sum_{i=0}^{M-1} \hat Z_i^* Z_i}{\frac{1}{M}\sum_{i=0}^{M-1} \abs*{\hat Z_i}^2}
		= \frac{\bar A_M}{\bar B_M},
	\end{align}
	and let
	\begin{align}
		C_\text{opt} = \frac{\mathbb{E}[\hat Z^* Z]}{\mathbb{E}[\abs*{\hat Z}^2]}
		= \frac{\mathbb{E}[A_1]}{\mathbb{E}[B_1]}
	\end{align}
	denote the theoretical MSE-optimal scalar calibration. 
	
	\begin{proof}[Proof of Corollary~\ref{thm:asymptotics}]
		\emph{(1) Consistency (almost sure convergence).}
		By assumption, the sequence $\{(A_i,B_i)\}_{i\geq 1}$ is strictly stationary (i) and ergodic (ii). Because the sequence is stationary and ergodic and integrable, sample averages of $A_i$ and $B_i$ converge almost surely to their expectations, the strong law of large numbers for such processes implies
		\begin{align}
			\bar A_M \xrightarrow{\text{a.s.}} \mathbb{E}[A_1],
			\qquad
			\bar B_M \xrightarrow{\text{a.s.}} \mathbb{E}[B_1],
			\label{eq:A_M-to-A_1}
		\end{align}
		as $M\to\infty$, provided the first moments are finite (iv). Ergodicity and suitable weak dependence is sufficient to ensure time averages converge to expectations.\\ 
		
		Assumption (iv) includes $\mathbb{E}[B_1] >0$. Since $\bar B_M \to \mathbb{E}[B_1]$ almost surely, once $\bar B_M$ is sufficiently close to $\mathbb{E}[B_1]$, it must be positive and bounded away from zero almost surely. Formally, there exists an $M_0$ such that for all $M\geq M_0$,
		\begin{align}
			\abs{\bar B_M - \mathbb{E}[B_1]} < \frac 12 \mathbb{E}[B_1].
		\end{align}
		Then for those $M$,
		\begin{align}
			\bar B_M \geq \frac 12 \mathbb{E}[B_1] > 0.
			\label{eq:B_M}
		\end{align}
		So, the ratio $\bar A_M / \bar B_M$ is eventually well-defined and stable.\\
		
		An algebraic bound can be found
		\begin{align}
			\hat C_M - C_\text{opt} = \frac{\bar A_M}{\bar B_M} - \frac{\mu_A}{\mu_B} = \frac{\mu_B (\bar A_M -\mu_A) - \mu_A (\bar B_M - \mu_B)}{\bar B_M \mu_B},
		\end{align}
		where $\mu_A=\mathbb{E}[A_1]$, $\mu_B = \mathbb{E}[B_1]$. The numerator converges to zero due to \eqref{eq:A_M-to-A_1}, and \eqref{eq:B_M} ensures a finite denominator. This proves the empirical calibration is a strongly consistent estimator of $C_\text{opt}$.\\
		
		\noindent
		\emph{(2) Asymptotic normality.}
		Under assumptions (i), (iii) and (v), a multivariate central limit theorem for weakly dependent stationary sequences applies to the real-valued vector process. (i) ensures stationary sequences, (iii) is the main assumption enabling the CLT with rate $\sqrt{M}$ and a well-defined covariance, (v) ensures the mixing CLT applies and the long-run covariance exists.
		\begin{align}
			\vec{Y}_i =
			\begin{pmatrix}
				\Re A_i \\
				\Im A_i \\
				B_i
			\end{pmatrix},
			\qquad
			\vec{\bar Y}_M = \frac{1}{M}\sum_{i=0}^{M-1} \vec{Y}_i
		\end{align}
		Weak dependence plus a $2+\delta$ moment condition allow a multivariate CLT for $\vec{\bar Y}_M$ with a long-run covariance $\mat{\Sigma}_Y$ that captures serial correlation, assuming the series converges. In the iid case, this reduces to $\mat{\Sigma}_Y = \operatorname{Cov}(\vec{Y}_1)$.
		\begin{align}
			\sqrt{M}\left(\vec{\bar Y}_M - \mathbb{E}[\vec{Y}_1]\right)
			\xrightarrow{d} \mathcal{N}(\vec{0},\mat{\Sigma}_Y),
			\qquad 
			\mat{\Sigma}_Y = \sum_{l\in \mathbb{Z}} \operatorname{Cov}(\vec{Y}_1, \vec{Y}_{1+l})
			\label{eq:distr-conv-Y_M}.
		\end{align}
		From the definitions,
		\begin{align}
			\bar A_M = \bar Y_{M,1} + \j \bar Y_{M,2}, \qquad \bar B_M = \bar Y_{M,3}.
		\end{align}
		Hence the estimator is not linear, but a smooth function of means,
		\begin{align}
			\hat C_M = \frac{\bar A_M}{\bar B_M} = \frac{\bar Y_{M,1} + \j \bar Y_{M,2}}{\bar Y_{M,3}}.
		\end{align}
		The delta method replaces a smooth nonlinear function by its first-order Taylor expansion around the mean, so asymptotically the estimator behaves like a linear function of a Gaussian vector~\cite{casellaStatisticalInference2002}. In this way, the empirical calibration $\hat C_M$ can be written as a smooth function of $\vec{\bar Y}_M$. For this purpose, we define a real mapping $f\!\!: \mathbb{R}^3 \to \mathbb{R}^2$ that outputs real and imaginary parts.
		\begin{align}
			f(x,y,b) = \begin{pmatrix}
				x/b\\
				y/b
			\end{pmatrix}
		\end{align}
		Then $f(\vec{\bar Y}_M) = (\Re \hat C_M, \Im \hat C_M)$. By assumption (iv) $\mu_B = \mathbb{E}[B_1]>0$ ensures that $f$ is differentiable in a neighbourhood of $\vec{\mu}$.
		\begin{align}
			f(\vec{\bar Y}_M) &\approx f(\vec{\mu}) + \mat{J}_f(\vec{\mu}) (\vec{\bar Y}_M - \vec{\mu})
		\end{align}
		Multiplying by $\sqrt{M}$, we obtain
		\begin{align}
			\sqrt{M} (f(\vec{\bar Y}_M) - f(\vec{\mu})) \approx \mat{J}_f (\vec{\mu}) \sqrt{M} (\vec{\bar Y}_M - \vec{\mu}).
		\end{align}
		Since the right-hand side converges to a Gaussian, so does the left-hand side.\\
		
		Computing the Jacobian $\mat{J}_f (\vec{\mu})$ with $\vec{\mu} = \mathbb{E} [\vec{Y}_1] = (\mu_x, \mu_y, \mu_B)$ yields
		\begin{align}
			\mat{J}_f (\mu) = \begin{pmatrix}
				\pdv{x} \frac xb & \pdv{y} \frac xb & \pdv{b} \frac xb\\
				\pdv{x} \frac yb & \pdv{y} \frac yb & \pdv{b} \frac yb
			\end{pmatrix} 
			= \begin{pmatrix}
				\frac{1}{\mu_B} & 0 & -\frac{\mu_x}{\mu_B^2}\\
				0 & \frac{1}{\mu_B} & -\frac{\mu_y}{\mu_B^2}
			\end{pmatrix}.
		\end{align}
		If \eqref{eq:distr-conv-Y_M} and $f$ is differentiable in $\vec{\mu}$,
		\begin{align}
			\sqrt{M} (f(\vec{\bar Y}_M) - f(\vec{\mu})) \xrightarrow{d} \mathcal{N}(\vec{0}, \mat{J}_f \mat{\Sigma}_Y \mat{J}_f^\mathsf{T}).
		\end{align}
		But $f(\vec{\mu}) = (\mu_x/\mu_B, \mu_y/\mu_B)$, which is exactly $(\Re C_\text{opt}, \Im C_\text{opt})$ because $\mu_x + \j \mu_y = \mathbb{E}[A_1]$ and $\mu_B= \mathbb{E}[B_1]$. Also $f(\vec{\bar Y}_M) = (\Re \hat C_M, \Im \hat C_M)$. Therefore,
		\begin{align}
			\sqrt{M} \begin{pmatrix}
				\Re\{\hat C_M - C_\text{opt}\}\\
				\Im\{\hat C_M - C_\text{opt}\}
			\end{pmatrix}
			\xrightarrow{d} \mathcal{N}_{\mathbb{C}} (\vec{0}, \mat{\Sigma}_C), \qquad \mat{\Sigma}_C = \mat{J}_f (\vec{\mu}) \mat{\Sigma}_Y \mat{J}_f (\vec{\mu})^\mathsf{T}
		\end{align}
		Finally, since a complex number is equivalent to its 2D real vector $(\Re\cdot, \Im\cdot)$, the above 2D real CLT can be summarised as
		\begin{align}
			\sqrt{M}\left(\hat C_M - C_\text{opt}\right)
			\xrightarrow{d}
			\mathcal{N}_{\mathbb{C}}(\vec{0},\mat{\Sigma}_C).
		\end{align}
		In particular, the variance of $\hat C_M$ decreases at the rate of $1/M$, and the real and imaginary parts of the calibration error are asymptotically jointly Gaussian.
	\end{proof}
	
	As mentioned previously, the assumptions of this derivation are not intended to model raw samples. Instead, the sequence $\{(Z_i,\hat Z_i)\}$ is formed from \emph{window-level endpoint estimates}, which aggregate many raw samples into a single complex-valued quantity. For band-limited, approximately stationary signals, such window-level features can be well approximated by stationary, weakly dependent processes over short lengths~\cite{blancoStationarityEEGSeries1995,florianDynamicSpectralAnalysis1995,fingelkurtsEditorialEEGPhenomenology2010,lawhernDetectingAlphaSpindle2013}. In this setting, the effective number of independent observations is determined by the number of usable windows $M$ rather than by the raw sampling rate. The above theorem states that the empirical calibration converges to the optimal Wiener gain and that its estimation error decays as $1/\sqrt{M}$ even in the presence of temporal correlations typical of EEG data.

	\subsection{Optimal scalar calibration for a single tone}
	\label{sec:scalar-calibration-proof}
	
	For the finite length cosine analysed in section~\ref{sec:error-analysis}, the ecHT endpoint obeys
	\begin{align}
		F = G_+ + G_- \e^{-\j 2\varphi_0},
	\end{align}
	where $G_+$ is the effective complex gain for the positive-frequency component and $G_-$ the effective gain (leakage) for the negative-frequency component. The initial phase $\varphi_0$ is modelled as random with $\varphi_0 \sim \mathcal{U}(-\pi, \pi)$ to represent an unknown phase offset.\\
	
	The average error over all possible initial phase offsets can be stated to get an expected error that does not depend on the arbitrary phase of the input signal. No closed-form solution exists for these integrals, but they can be evaluated numerically.
	\begin{align}
		\overline{\varepsilon}_\theta &= \mathbb{E} [\varepsilon_\theta] = \frac{1}{2\pi} \int_{-\pi}^\pi \arg F \dd{\varphi_0}\\
		\overline{\varepsilon}_A &= \mathbb{E} [\varepsilon_A] = \frac{1}{2\pi} \int_{-\pi}^\pi \abs{F}-1 \dd{\varphi_0}		
	\end{align}
	The mean-squared error can also be stated directly since $\mathbb{E} [\e^{\pm \j 2\varphi_0}] = 0$.
	\begin{align}
		\begin{split}
			J &= \mathbb{E} \abs{F-1}^2\\
			&= \mathbb{E} \left[ \abs{G_+-1}^2 + \abs{G_-}^2 + (G_+-1)G_-^* \e^{\j 2\varphi_0} + (G_+-1)^*G_- \e^{-\j 2\varphi_0} \right]\\
			&= \abs{G_+-1}^2 + \abs{G_-}^2
		\end{split}
	\end{align}
	$\abs{G_+-1}^2$ contributes systematic distortion of the main (positive) component and $\abs{G_-}^2$ contributes leakage from the negative component. The general result from Section~\ref{sec:calibration-general-signal} is specialised in this section.\\
	
	Under the assumptions above, the MSE-optimal scalar calibration factor 	\eqref{eq:Copt-general} specialises to
	\begin{align}
		C_\text{opt}
		=
		\frac{G_+^*}{\abs{G_+}^2 + \abs{G_-}^2}.
	\end{align}
	The corresponding minimal mean-square error is
	\begin{align}
		J\left(C_\text{opt}\right)
		= \frac{\abs{G_-}^2}{\abs{G_+}^2 + \abs{G_-}^2}.
	\end{align}
	
	\begin{proof}[Proof of Theorem~\ref{thm:single-tone-calibration}]
		Since $\hat{Z} = F Z$ and $\abs{Z} = 1$, it follows that
		\begin{align}
			\hat{Z}^* Z = F^*,
			\qquad
			\abs*{\hat{Z}}^2 = \abs{F}^2.
		\end{align}
		Taking expectations with respect to $\varphi_0$,
		\begin{align}
			\mathbb{E}[\hat{Z}^* Z]
			= \mathbb{E}[F^*]
			= \mathbb{E}[G_+^* + G_-^* \e^{\j2\varphi_0}]
			= G_+^*,
		\end{align}
		because $\mathbb{E}[e^{\pm\j2\varphi_0}] = 0$ for $\varphi_0 \sim \mathcal{U}(-\pi,\pi)$.
		Furthermore,
		\begin{align}
			\label{eq:F=G++G-}
			\begin{split}
				\mathbb{E}[\abs*{\hat{Z}}^2]
				&= \mathbb{E}[\abs{F}^2]
				= \mathbb{E}[\abs{G_+ + G_- \e^{-\j 2\varphi_0}}^2]  \\
				&= \mathbb{E}\left[\abs{G_+}^2 + \abs{G_-}^2 
				+ G_+ G_-^* \e^{\j2\varphi_0} + G_+^* G_- \e^{-\j2\varphi_0}\right]\\
				&= \abs{G_+}^2 + \abs{G_-}^2.
			\end{split}
		\end{align}
		Insertion into \eqref{eq:Copt-general} yields
		\begin{align}
			C_\text{opt}
			= \frac{\mathbb{E}[\hat{Z}^* Z]}{\mathbb{E}[\abs*{\hat{Z}}^2]}
			= \frac{G_+^*}{\abs{G_+}^2 + \abs{G_-}^2},
		\end{align}
		which is \eqref{eq:Copt-single-tone}. For the minimal MSE, \eqref{eq:Jmin-general}
		is used with $\mathbb{E}[\abs{Z}^2] = 1$:
		\begin{align}
			\begin{split}
				J\left(C_\text{opt}\right)
				&= 1 - \frac{\abs*{\mathbb{E}[\hat{Z}^* Z]}^2}{\mathbb{E}[\abs*{\hat{Z}}^2]}  \\
				&= 1 - \frac{\abs{G_+}^2}{\abs{G_+}^2 + \abs{G_-}^2}
				= \frac{\abs{G_-}^2}{\abs{G_+}^2 + \abs{G_-}^2},
			\end{split}
		\end{align}
		which proves \eqref{eq:Jmin-single-tone}.
	\end{proof}
	
	Thus, for a single tone the best scalar calibration depends only on the	positive-frequency gain $G_+$ and the negative-frequency leakage $G_-$, and the minimal achievable MSE is exactly the fraction of leakage power $\abs{G_-}^2$ in the total ecHT endpoint power $\abs{G_+}^2 + \abs{G_-}^2$. The coefficient $C_\text{opt}$ coincides with the scalar Wiener filter that produces the linear minimum-variance estimate of $Z$ from the observation $\hat{Z}$ in the single-tone setting.\\
	
	The corresponding correlation coefficient and deterministic bounds on phase and	amplitude error can also be expressed explicitly.
	\begin{align}
		\rho_{Z\hat{Z}}
		= \frac{\mathbb{E}[\hat{Z}^* Z]}
		{\sqrt{\mathbb{E}[\abs*{\hat{Z}}^2]\;\mathbb{E}[\abs{Z}^2]}}
		= \frac{G_+^*}{\sqrt{\abs{G_+}^2 + \abs{G_-}^2}},
	\end{align}
	so that
	\begin{align}
		\abs{\rho_{Z\hat{Z}}}^2 = \frac{\abs{G_+}^2}{\abs{G_+}^2 + \abs{G_-}^2},
	\end{align}
	and the minimal MSE \eqref{eq:Jmin-single-tone} can be written as
	\begin{align}
		J\left(C_\text{opt}\right)
		= 1 - \abs{\rho_{Z\hat{Z}}}^2.
	\end{align}
	
	Thus, the leakage power $\abs{G_-}^2$ directly reduces the squared correlation between the ecHT endpoint and the ideal analytic endpoint, and therefore limits the minimal achievable MSE.

	\section{Discussion of argument function} \label{sec:phase1 vs phase2}
	
	There are two natural but not equivalent ways to define the instantaneous phase from the analytic signal.
	
	\subsection{Method 1: Phase of the analytic signal}
	
	The first definition takes the phase directly as the argument of the complex analytic signal.
	\begin{align}
		z(n) &= x(n) + \j \mathcal{H}x(n),\\
		\theta_1(n) &= \arg z(n) = \arctantwo \left(\Im z(n),\, \Re z(n)\right).
	\end{align}
	
	The endpoint value can be written as
	\begin{align}
		\hat{z}_\text{end}= z_\text{end}\,F
	\end{align}
	where
	\begin{align}
		z_\text{end} = \e^{\j(\omega_0(N-1)+\varphi_0)}
	\end{align}
	is the ideal analytic signal of the cosine at the endpoint and $F$ is the effective complex gain that collects the effects of finite windowing, analytic masking and bandpass filtering (see Section~\ref{sec:error-analysis}). Hence
	\begin{align}
		\begin{split}
			\theta_1(N-1)
			&= \arg \hat{z}_\text{end}\\
			&= \arg \{z_\text{end}F\}\\
			&= \omega_0(N-1) + \varphi_0 + \arg F.
		\end{split}
	\end{align}
	Thus, Method~1 preserves the linear phase progression $\omega_0 n + \varphi_0$ of the tone and only adds a phase error $\arg F$ which decomposes into a calibratable bias $\arg G_+$ and a leakage-induced ripple bounded by $\arcsin r$.
	
	\subsection{Method 2: Phase from the ratio to the original signal}
	
	A second definition, which has also appeared in the ecHT literature, replaces the real part of the analytic signal in the denominator by the original signal	$x$.
	\begin{align}
		\theta_2(n)
		&= \arctantwo \left(\Im z(n),\, x(n)\right)
	\end{align}
	For an ideal analytic signal, $\Re z(n) = x(n)$, so $\theta_1(n) = \theta_2(n)$. However, finite windowing and bandpass filtering distort the real and imaginary parts of $z(n)$ relative to the original signal. Using the endpoint representation
	\begin{align}
		\hat{z}_\text{end} &= z_\text{end}F\\
		z_\text{end} &= \cos\psi + \j \sin\psi\\
		\psi &= \omega_0(N-1)+\varphi_0.
	\end{align}
	Multiplying out the product using $F = F_\text{R} + \j F_\text{I}$
	\begin{align}
		\begin{split}
			\hat{z}_\text{end}
			&= (\cos\psi + \j\sin\psi)(F_R + \j F_I) \\
			&= F_R\cos\psi
			+ \j F_R\sin\psi
			+ \j F_I\cos\psi
			- F_I\sin\psi,
		\end{split}
	\end{align}
	one obtains
	\begin{align}
		\Re \hat{z}_\text{end}
		&= F_R\cos\psi - F_I\sin\psi,\\
		\Im \hat{z}_\text{end}
		&= F_R\sin\psi + F_I\cos\psi.
	\end{align}
	Using $\Re z_\text{end} = \cos\psi$, Method~2 yields
	\begin{align}
		\tan\theta_2(N-1)
		&= \frac{\Im \hat{z}_\text{end}}{x_\text{end}}
		= \frac{F_R \sin\psi + F_I \cos\psi}{\cos\psi}
		= F_R \tan\psi + F_I.
	\end{align}
	Consequently,
	\begin{align}
		\theta_2(N-1)
		= \arctantwo \big(F_R \sin\psi + F_I \cos \psi,\, \cos\psi\big),
	\end{align}
	which is a nonlinear function of $n$ (through $\psi = \omega_0(N-1) +	\varphi_0$) and of the initial phase $\varphi_0$. This shows that the effective ecHT gain $F$ causes a rotation and non-uniform scaling of the cosine-sine quadrature pair. In the ideal case	$F = 1$, the real and imaginary components recover the standard analytic signal representation.\\
	
	In contrast to Method~1, the phase evolution is no longer strictly linear, but is warped by $F$ which includes filter and window effects. For this reason, Method~2 is not recommended for ecHT phase estimation.
	
	% \printbibliography

\end{document}